\documentclass[twocolumn]{article}

\usepackage{graphics} 
\usepackage{epsfig} 
\usepackage{textcomp}
\usepackage{times} 
\usepackage{amsmath} 
\usepackage{amsthm}
\refstepcounter{equation} 
\usepackage{amssymb}  
\usepackage{mathtools}
\usepackage{adjustbox}
\usepackage{xfrac }

\usepackage{pgf, tikz}
\usetikzlibrary{shapes.geometric}
\newtheorem{lemma}{Lemma}
\newtheorem{theorem}{Theorem}
\newtheorem{remark}{Remark}

\newtheorem{assumption}{Assumption}
\newtheorem{definition}{Definition}

\usepackage{color,soul} 
\soulregister\cite7
\soulregister\ref7
\soulregister\pageref7

\begin{document}

\title{
Plug-and-Play Model Predictive Control for Load Shaping and Voltage Control in Smart Grids}

\author{{ Caroline Le Floch$^{\text{a}}$, Somil Bansal$^{\text{b}}$, 
Claire J. Tomlin$^{\text{b}}$, Scott Moura$^{\text{a}}$, Melanie Zeilinger$^{\text{c}}$}
\thanks{This work was funded in part by the Siebel Energy Institute, the Swiss National Science Foundation and the National Science Foundation within the Division of Electrical, Communications and Cyber Systems under Grant 1408107.}
         \thanks{$^{\text{a}}$Civil and Environmental Engineering,
				University of California, Berkeley,
				United States, {\small \tt caroline.le-floch@berkeley.edu}}
         \thanks{$^{\text{b}}$Electrical Engineering and Computer Sciences,
				University of California, Berkeley,
				United States}
	   \thanks{$^{\text{c}}$Mechanical and Process Engineering,
				ETH Zurich,
				Switzerland}
	 }


\maketitle


\begin{abstract}
This paper presents a predictive controller for handling plug-and-play (P\&P) charging requests of flexible loads in a distribution system. We define two types of flexible loads: (i) deferrable loads that have a fixed power profile but can be deferred in time and (ii) shapeable loads that have flexible power profiles but fixed energy requests, such as Plug-in Electric Vehicles (PEVs). The proposed method uses a hierarchical control scheme based on a model 
predictive control (MPC) formulation for minimizing the global system cost. The first stage computes a reachable 
reference that trades off deviation from the nominal voltage with the required generation control. The second stage uses a price-based objective to aggregate flexible loads and provide load shaping services, while satisfying system constraints and users' preferences at all times.
It is shown that the proposed controller is recursively feasible under specific conditions, i.e. the flexible load demands are satisfied  and bus voltages remain within the desired limits. Finally, the proposed scheme is illustrated on a 55 bus radial distribution network.
\end{abstract}

\section{INTRODUCTION} \label{section:intro}
The development of smart meters has led to the modernization of distribution networks by enabling real-time bidirectional communication \cite{wang2011survey, fan2013smart}. In this context of smart grids, demand side management allows system operators to control the energy consumption at the household level, offering new opportunities to improve the reliability, efficiency and sustainability of the grid \cite{palensky2011demand, mohsenian2010autonomous, Gellings2009}. In particular, automated load shifting is expected to play a key role in stabilizing the grid in the case of high penetration of Plug in Electric vehicles (PEVs) and uncontrollable renewable sources, such as photovoltaic solar energy \cite{LeFloch2015Distributed, Lund2008}. 

PEVs provide a compelling opportunity for supplying demand-side management services in the smart grid. For example, a vehicle-to-grid (V2G) capable PEV can communicate with the grid, can store energy, and can transfer energy back to the electric grid when required  \cite{richardson2013electric, kempton2005vehicle, Langton2013}. Other flexible commercial and residential loads such as thermostats, controllable lighting or dishwashers can be deferred to adapt their consumption according to the distribution grid constraints. In this article, we consider two types of services for controlling these flexible loads. On a localized and short-time scale, the controller provides voltage regulation at distribution buses. On a longer time scale, the controller aggregates and schedules flexible loads to reduce peak consumption and shape daily load curves. The goal of this paper is to develop a control scheme that integrates flexible loads in the existing distribution network, provides local and aggregated grid services and satisfies users$'$ requirements.

There are two key challenges in designing such a controller. First, it should be able to handle variations in the number of connected loads, i.e. plugging in and out operations. In a real scenario, a user can request to connect or disconnect any load at any desired time and bus. Modeling all the possible requests results in a very large scale and uncertain system, which can be studied with distributed algorithms and model predictive control (MPC)  \cite{delfino2014multilevel}. In contrast, we model only the static state of the network and deal with plug-in requests when new loads require supply. This will modify the overall distribution system, and the current load schedule may be infeasible and/or unstable for the modified system, requiring a controller redesign. Second, the controller should consider two different objectives and time scales: local voltage regulation and aggregate load shaping. Previous work has proposed multilevel and multi-horizon approaches \cite{delfino2014multilevel, shaaban2014real, lopes2011integration} and decentralized protocols \cite{mou2015decentralized} to address this issue. 

In this work, our goal is to schedule loads in real time to provide load shaping services while satisfying voltage constraints at each bus and time step. We define a two stage plug-and-play model predictive controller.
While the main focus of MPC so far has been on the control of networks with constant topology, the concept of plug and play (P\&P) MPC \cite{riverso2015plug, Melanie, Stoustrup2009}  considers network changes by subsystems that want to join or leave the network, while ensuring feasibility and stability of the global system. By providing an automatic redesign of the control laws in response to changing network conditions,
P\&P MPC is an attractive scheme for modern control systems of increasing complexity. We provide a P\&P framework
to deal with the connection and disconnection of loads from the grid, which requires an online feasibility handling as introduced in \cite{bansal2014plug}. A procedure for updating the controller together with a transition scheme is proposed, which prepares the system for the requested modifications. A schematic representation of the protocol is given in Fig. \ref{fig:schema_general}.

\begin{figure}[tb]
	\centering
    \includegraphics[page=2,trim = 21mm 18mm 10mm 37mm, clip, width = 0.5\textwidth]{Schema_gen.pdf}
    \includegraphics[trim = 2mm 16mm 0mm 0mm, clip, width = 0.5\textwidth]{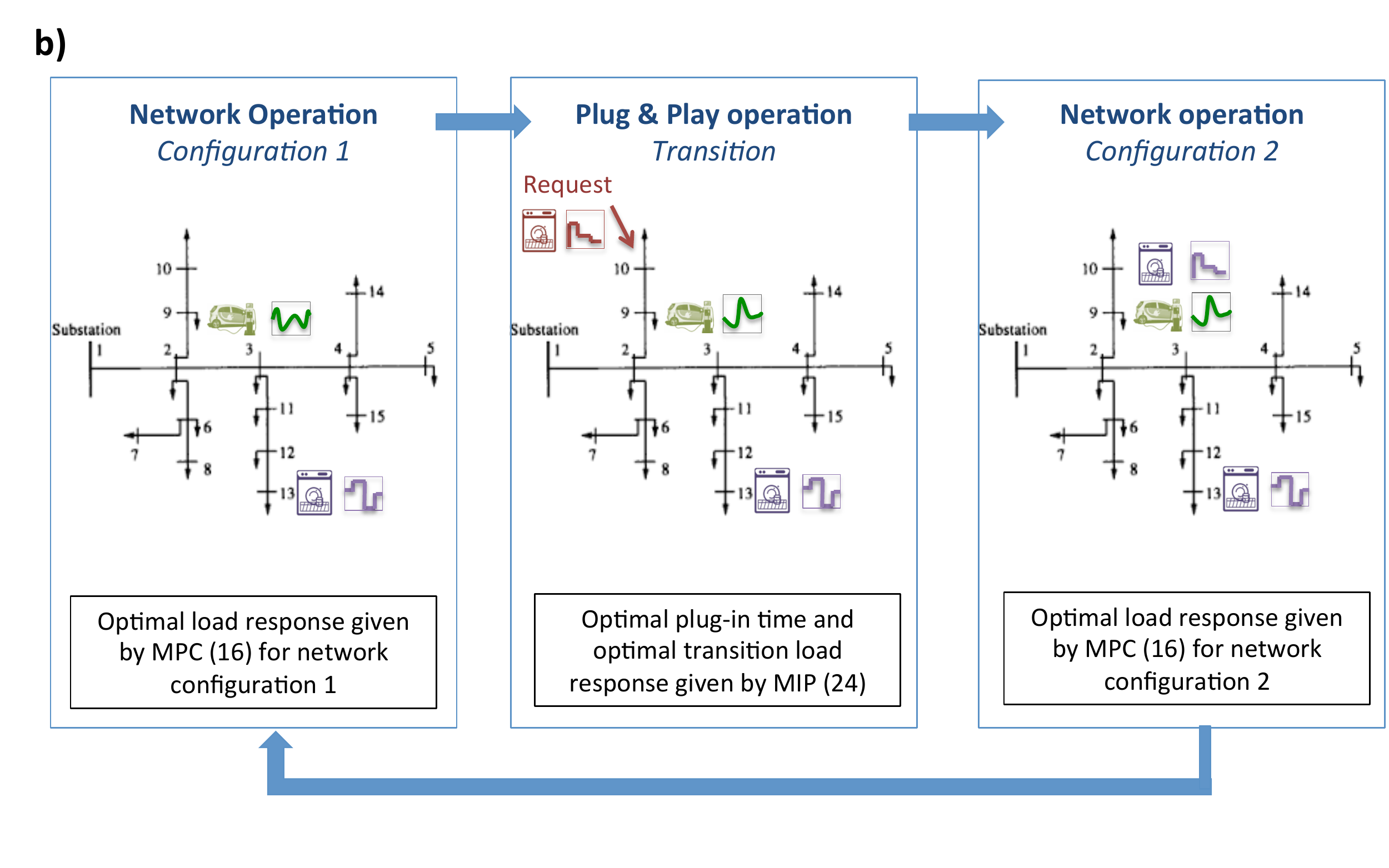}
	\centering
    \caption{Schematic representation of the protocol: (a) loads plugging in and out of the distribution network and (b) iterative process between network operations and plug \& play operations.}
	\label{fig:schema_general}
\end{figure} 
To summarize, the main contributions of this paper are:
\begin{itemize}
\item It presents a novel P\&P MPC scheme that optimally schedule loads to be connected and shape current loads while satisfying network constraints at all times. We model power flows in the distribution network using the Second Order Cone relaxation of the DistFlow equations (see \cite{farivar2013branch}) and improve the accuracy of the model compared to the linear approximation used in \cite{bansal2014plug}.
\item The controller is applicable to different types of loads. In particular, we define two types of loads: deferrable loads that can be delayed but have a fixed profile, and shapeable loads that have a flexible profile but need a fixed amount of energy. The proposed algorithm satisfies users' requirements by ensuring that every flexible load reaches the desired energy level at the desired time and any deferrable load demand is met before its deadline.
\end{itemize}

The paper is organized as follows: Section \ref{secn:prelim} introduces the system model. In Section \ref{secn:controller}, control objectives are defined and the hierarchical MPC controller is proposed. Section \ref{secn:P&P} provides an extension for handling plug and play requests. Section \ref{secn:example} presents numerical simulations demonstrating the advantages of the proposed control scheme and Section \ref{secn:conclusion} provides concluding remarks. 

\section{PRELIMINARIES} \label{secn:prelim}
In this section we introduce the different elements of the network, including deferrable and shapeable loads, battery banks and capacitors. We model the overall system as a constrained dynamic system with linear dynamics and SOC constraints. 
\subsection{Load modeling} \label{secn:load}
Let $e(t)$ be an initial energy profile. The role of the scheduling operator 
is to generate another energy profile $\widetilde{e}=\tau (e)$.
We distinguish 3 types of loads, which have different possible maps $\tau$ (see Fig. \ref{fig:schema_general}a):
\begin{itemize}
  \item \emph{Fixed loads} do not participate in demand response and cannot be controlled
  : $\widetilde{e}^f=e^f$.
  \item \emph{Deferrable loads} can be delayed but have a fixed load profile. In 
  this case $\widetilde{e}^{def}(t)=e^{def}(t-d)$ with $d$ bounded by a constraint on the maximum allowable
delay. This includes PEVs with constant charging rate.
\item \emph{Shapeable loads} have a flexible profile but need a fixed amount of energy 
in a fixed time period T:
 $\int_T  \widetilde{e}^{shp}= \int_T  e^{shp} dt$. This includes PEVs with continuous 
 charging rate.
\end{itemize}
In the following sections, we model the load dynamics.

\subsubsection{Deferrable loads}
{Consider a deferrable load with a power profile $c^{def}_0 >0$.} As soon as the deferrable load plugs in, its power profile  $c^{def}$ becomes deterministic:
\begin{subequations}
\begin{eqnarray} \label{eqn:def0}
  e^{def}(k+1)&=&e^{def}(k) + \eta^{def} \Delta T c^{def}(k)\\
c^{def}(k)&=&c^{def}_0(k)
\end{eqnarray}
\end{subequations}
where $\Delta T$ is the sampling time and $\eta^{def}<1$ is an efficiency coefficient.
\begin{remark} \label{defferable}
  The only control on these loads comes from P\&P operations, which determines when to plug-in the load. After it is connected, the load is deterministic and can be considered as a fixed load. 
  \end{remark}
 {~}\\
  
\subsubsection{Shapeable loads}$\mbox{  }$\\
Power at shapeable loads $c^{shp}$ can take values in a continuous range $[0, 
c^{shp}_{max}]$. We define $e^{shp}_{low}$, $e^{shp}_{max}$ and $e^{shp}_{des}$ as the physical lower limit, physical upper limit and desired State Of Charge (SOC) of the load, respectively. 
The dynamics of shapeable loads are given by:
\begin{subequations}
 \begin{eqnarray} \label{eqn:shp}
  e^{shp}(k+1)&=&e^{shp}(k) + \eta \Delta T c^{shp}(k)\\
  e^{shp}_{min}(k) \leq & e^{shp} &\leq e^{shp}_{max}\\
0\leq & c^{shp} &\leq c^{shp}_{max} \label{eqn:shp2}
\end{eqnarray}
\end{subequations}
where $\eta<1$ is an efficiency coefficient.
We assume that the load has to be fully charged
by time $k^{out}$: $e^{shp}(k^{out}\geq e^{shp}_{des} $. We ensure that this is feasible by imposing the time varying
lower bound constraint on the SOC value:
\begin{equation}
  e^{shp}_{min}(k)= e^{shp}_{des} - (k^{out}-k )c^{shp}_{max} \Delta T \eta~~~~~~~~\forall k\leq k^{out}.
  \label{eq:SOC_feas}
\end{equation}
We can combine Eq (\ref{eq:SOC_feas}) with the physical lower limit $e^{shp}_{low}$ and reformulate:
\begin{equation} \label{eq:shapeableconstraint}
e^{shp}_{min}(k)=\max [e^{shp}_{des}\mbox{ - }\max(0,\small{(k^{out}\mbox{ - }k )}{c^{shp}_{max}}{\eta \Delta T}), e^{shp}_{low}].
\end{equation}
In the remainder of this paper we use the notation:
\begin{center}
$ 1_{k< k^{out}}=  \begin{cases}
    1  &  \text{if $ k<k^{out}$} \\
    0  & \text{otherwise} \\
  \end{cases}
$
\end{center}

\subsection{Battery banks}
We model on-site batteries with a linear state space model:
\begin{subequations}
\begin{eqnarray} 
 \label{eqn:def}
e^{bat}(k+1)&=&e^{bat}(k) + \Delta T p^{bat}\\
e^{bat}_{low} \leq & e^{bat} &\leq e^{bat}_{max}\\
 p^{bat}_{min}  \leq & p^{bat} &\leq  p^{bat}_{max} 
\end{eqnarray}
\end{subequations}
where $e^{bat}_{low}$, $e^{bat}_{max}$ are the fixed physical lower and upper limits of the battery's SOC, and  $p^{bat}_{min}$, $p^{bat}_{max}$ are the minimum and maximum power. This model assumes perfect battery efficiency, which is a simplified approximation of the battery dynamics, but frequently used in the power system literature to formulate the overall system as a linear dynamical system, and to simplify the resulting control scheme (\cite{he2012optimal}, \cite{bansal2014plug}, 
 \cite{Hu2014Coordinated}).
\subsection{Network Model}

We consider a radial distribution network, which is a structure commonly used in the power systems literature. To characterize the power flow in this network we adopt the DistFlow equations first introduced in \cite{Baran1989} and the notation introduced in \cite{Farivar2013}, restated here for completeness. 
\begin{table}[ht]

\caption{Variables for a radial distribution network}
\centering
\begin{adjustbox}{max width=0.5\textwidth}
\begin{tabular}{|lcl|}
\hline
$\mathcal{N}$ &  & Set of buses, $\mathcal{N} := \{1,\ldots,n\}$ \\
$\mathcal{L}$ & & Set of lines between the buses in $\mathcal{N}$ \\ 
$\mathcal{L}_i$ &  & Set of lines connecting bus $0$ to bus $i$ \\ 
$p_i^l$ & & Real power consumption by fixed loads at bus $i$\\
$p_i^{bat}$ & & Real power consumption by battery banks at bus $i$\\
$p_i^{shp}$ & & Real power consumption by shapeable loads at bus $i$\\
$p_i^{def}$ & & Real power deferrable loads at bus $i$\\
$q_i^l, q_i^g$ & & Reactive power consumption and generation at bus $i$\\
$r_{ij}, x_{ij}$ & & Resistance and reactance of line $(i,j) \in \mathcal{L}$\\
$P_{ij}, Q_{ij}$ & & Real and reactive power flows from bus $i$ to $j$\\
$v_i$ & & Voltage magnitude at bus $i$\\
$l_{ij}$ & & Squared magnitude of complex current from bus $i$ to $j$\\
$M_i$ & & Number of shapeable loads connected at bus $i$\\
\hline
\end{tabular}
\label{table:notations}
\end{adjustbox}
\end{table}

The power flow equations for a radial distribution network can be written as the following DistFlow equations \cite{Baran1989_2}:
%
\begin{subequations} \label{eqn:orig_PF}
\begin{equation}\label{eqn:orig_PF1}
P_{ij} =   p_{j}^{l} + p_{j}^{bat} +p_{j}^{def}+p_{j}^{shp}+ r_{ij}l_{ij} + \sum\limits_{k:(j,k) \in \mathcal{L}} P_{jk}
\end{equation}
\begin{equation}\label{eqn:orig_PF2}
Q_{ij}   =  q_{j}^{l}  - q_{j}^{g} + x_{ij}l_{ij} + \sum\limits_{k:(j,k) \in \mathcal{L}} Q_{jk}
\end{equation}
\begin{equation}\label{eqn:orig_PF3}
v_{j}^2  = v_{i}^2 + (r_{ij}^2 + x_{ij}^2)l_{ij} - 2(r_{ij}P_{ij} + x_{ij}Q_{ij})
\end{equation}
\begin{equation} \label{eq:non_convex}
l_{ij}v_{i}^2  =  P_{ij}^2 + Q_{ij}^2
\end{equation}
\begin{equation*}
\forall j \in \mathcal{N}\backslash\lbrace{1}\rbrace \mbox{, and } (i,j) \in \mathcal{L}
\end{equation*}
\end{subequations}
 where $P_{ij}$, $Q_{ij}$, $v_j$ and $l_{ij}$ are defined in Table \ref{table:notations}. 
Because the above formulation is non-convex, we use the Second Order Cone relaxation (SOCP) defined in \cite{farivar2013branch}, where equation (\ref{eq:non_convex}) is relaxed as follows:
\begin{equation} \label{eq:convex}
l_{ij}  \geq \frac{P_{ij}^2 + Q_{ij}^2} {v_{i}^2}.
\end{equation}

 The variables 
are the reactive power generation input (column) vector $q^g := (q_1^g, \ldots, q_n^g) \in \mathbb{R}^n$, the 
battery input vector $p^{bat} := (p_1^{bat}, \ldots, p_n^{bat}) \in \mathbb{R}^n$ and the shapeable, deferrable and fixed 
loads: $p^{shp}$, $p^{def}$, $p^l$ $\in \mathbb{R}^n$, where $p_i^{shp} \in \mathbb{R}^+$ and $p_i^{def} \in \mathbb{R}^+$ denote the net shapeable loads and  
net deferrable loads charging at bus $i$, respectively. We denote $M_i^{shp}$ and $M_i^{def}$ the number of shapeable and deferrable
loads connected at bus i, respectively. These values can vary over time due to plugging and unplugging operations. This 
relates with the notation in Section \ref{secn:load} as follows:
\begin{eqnarray}
  p_i^{shp} &=& \sum\limits_{j=1}^{M_i^{shp}} c^{shp}_{j}\\
   p_i^{def} &= &\sum\limits_{j=1}^{M_i^{def}} c^{def}_{j}
\end{eqnarray}

We assume that the substation voltage $v_0$ is given and is constant. Furthermore, load profiles $p^l$ and $q^l$ 
are time-varying but their 24 hour forecast is assumed to be given.

\subsection{Network Constraints} 

Depending on the load, bus voltages can fluctuate significantly. For reliable operation of the distribution 
network it is required to maintain the bus voltages $v$ within a tight range around the nominal value $v_{nom}$ 
at all times (generally $5\% $ deviation): 
\begin{equation} \label{eqn:vol_constraint} v_{nom}-\Delta v_{max} \leq v \leq v_{nom}+\Delta v_{max}. \end{equation}
{We define the variable $\nu_i=v_i^2$ and write condition (\ref{eqn:vol_constraint}) as:}
\begin{equation} \label{eqn:vol constraint2} \nu_{min} \leq \nu_i \leq \nu_{max}. \end{equation}

In addition, there are inherent physical limitations on the capacitor control input, which is limited to:
\begin{equation} \label{eqn:gen constraint} q_{min} \leq q^g \leq q_{max}. \end{equation}

\subsection{Dynamic System} \label{sec:DynSys}
In this section, we represent the overall system as a constrained dynamic system with SOCs as states. 
Recalling that $p_i^{shp} = \sum\limits_{j=1}^{M^{shp}_i} c^{shp}_{j}$, we can write
$p^{shp}=K^{shp}u^{shp}$
where $u^{shp} := (c^{shp}_1,\ldots,c^{shp}_{M^{shp}_1},\ldots,c^{shp}_{M^{shp}})^T \in \mathbb{R}^{M^{shp}}$, 
and $M^{shp}$ is the total number of shapeable loads connected to the grid, i.e. $M^{shp} = \sum\limits_{j=1}^n M^{shp}_j$. Matrix $K^{shp} \in \mathbb{R}^{n \times M^{shp}}$ 
is defined such that:
\begin{equation*}
  K^{shp}_{ij}=  \begin{cases}
    1       &  \text{if shapeable load
 j is connected to bus i} \\
    0  & \text{otherwise } \\
  \end{cases}
  \end{equation*}

The overall system model is described as follows:
\begin{subequations}
 \label{eqn:sys1}
 \begin{align}\label{eqn:sys1-1}
&x(k+1)  =  Ax(k) + Bu(k)\\
\label{eqn:sys2}
&(x(k), u(k)) \in \mathcal{Z}_k 
\end{align}
\end{subequations}

where
\\
\begin{equation*}
\begin{array}{c}
 x = \begin{bmatrix} x_1, x_2 \end{bmatrix} ^T= \begin{bmatrix} (e^{shp}_1,\ldots,{e}^{shp}_{M^{shp}}) ,(e^{bat}_1\ldots,e^{bat}_n)\end{bmatrix}^T \\
 u = \begin{bmatrix} q^{g} ~ u^{shp} ~ p^{bat} \end{bmatrix}_{}^T 
\end{array}
\end{equation*}

\begin{equation*}
\begin{array}{cc}
A = I, & B = \begin{bmatrix} 0 &  \eta \Delta T & 0 \\ 0 & 0 &  \Delta T \end{bmatrix}\\

\end{array}
\end{equation*}
$
\begin{aligned}
\mathcal{Z}_k =\{ & (x(k), u(k), P_{ij}(k), Q_{ij}(k), \nu_i(k), l_{ij}(k)) :\\
 &P_{ij}(k) =   (p_{j}^{l}(k) \mbox{+}p_{j}^{bat}(k) \mbox{+}p_{j}^{def}(k)\mbox{+}p_{j}^{shp}(k))\\
 &~~~~~~~~~~~ + r_{ij}l_{ij} (k)\mbox{+}\sum_{l:(j,l) \in \mathcal{L}} P_{jl}(k),\\
 &Q_{ij}(k)   =  q_{j}^{l}  - q_{j}^{g} + x_{ij}l_{ij}(k) + \sum\limits_{l:(j,l) \in \mathcal{L}} Q_{jl}(k),\\
 &\nu_{j} (k) = \nu_{i}(k) + (r_{ij}^2 \mbox{+} x_{ij}^2)l_{ij}(k) \\
 &~~~~~~~~~~- 2[r_{ij}P_{ij}(k) \mbox{+ } x_{ij}Q_{ij}(k)],\\
 & l_{ij}  \geq \frac{P_{ij}(k)^2 + Q_{ij}(k)^2} {\nu_{i}(k)},\\
& \nu_{min}\leq\nu_i(k)\leq\nu_{max},\\
& e_{min}(k) \leq x(k) \leq e_{max}, \\
& p^{bat}_{min}  \leq  p^{bat} \leq  p^{bat}_{max}  \\
&q_{min} \leq q^g(k) \leq q_{max}, \\ 
&0 \leq u^{shp}(k) \leq c^{shp}_{max} \}.\\
\end{aligned} 
$

\section{Controller Design} \label{secn:controller}

In this section, we design a controller that captures three control objectives: 
\begin{itemize}
  \item  \emph{Peak reduction}: Smooth the aggregated power profile.
  \item\emph{User satisfaction}: Provide the desired energy to shapeable and deferrable loads.
  \item \emph{Voltage control}: Ensure that voltage deviation from nominal voltage remains within bounds. 
  
\end{itemize}

For the purpose of this section we assume that the number of loads connected to the grid is constant, i.e. no new loads are connected to 
or disconnected from the system. 
Plug and play connections are introduced in Section \ref{secn:P&P}. \\

\subsection{{Stage 1: Feasible reference}}
In the remainder of this paper, Equation (\ref{eqn:sys1}) with $p^{def}=0$, $u^{shp}=0$ refers to the 
dynamics with no deferrable and no shapeable loads.  We assume that the system has a feasible trajectory and has the following periodic property, with period $N_r$:
\begin{assumption} 
\label{as:feasible}
  There exists an initial value ${\hat{x}}_0 $ and a sequence of control inputs $\hat{u}(k)$,
  such that the corresponding sequence of states $\hat{x}(k)$ according to dynamics (\ref{eqn:sys1-1})  with $p^{def}=0$, $u^{shp}=0$ satisfies the constraints in (\ref{eqn:sys2}), i.e. $(\hat{x}(k), \hat{u}(k)) \in \mathcal{Z}_k$  for all $k\in \{0,...,N_r-1\}$.
\end{assumption}
\begin{assumption} 
\label{as:periodic}
  If problem $(\ref{eqn:sys1})$ with $p^{def}=0$, $u^{shp}=0$, $x(k)=x_0$ is 
  feasible at time $k$, then problem (\ref{eqn:sys1}) with $p^{def}=0$, $u^{shp}=0$, $x(k+N_r)=x_0$
  is feasible at time $k+N_r$.
\end{assumption}
In practice, Assumption \ref{as:feasible} means that the traditional control 
devices (battery banks and capacitors) are selected according to the 
traditional fixed loads $p^l$. When new loads, such as PEVs, are not plugged-in, 
traditional control devices have enough flexibility
to regulate voltage. Assumption \ref{as:periodic} means that if the 
problem is feasible at time $k$, then the problem with same initial state is feasible 
at time $k+N_r$, where the period $N_r$ is typically a day.
 Deferrable and shapeable loads increase power demand 
and voltage drop, which requires extra control capacity until the loads are fully 
charged. In this case, the problem with extra loads may not be feasible, requiring to solve the problem in a hierarchical way. First an optimal 
solution is computed for the system without extra loads, second this reference signal is used to formulate a model predictive controller
with the augmented system \cite{Bansai2014Plug}.
The optimization problem for computing the solution
with only fixed loads $(\hat{{x}}(k); \hat{{u}}(k))$ is referred to as \emph{stage-1} in this 
paper: 
\begin{equation*}
V_1(\tilde{x},\tilde{u}) = \sum\limits_{i=0}^{N_r-1} ||\tilde{u}(i)||^2_{T_1}+ ||\nu(i)-\nu_{nom}||^2_{T_2}.\\
\end{equation*}

\begin{subequations} \label{eq:stage1}
\begin{align} 
&(\hat{x}(k),\hat{u}(k)) := \operatornamewithlimits{argmin}\limits_{\tilde{x},\tilde{u}} 
V_1(\tilde{x},\tilde{u})\\
&~~~~~~\mbox{s.t.}~~ \tilde{x}(i+1)  =  A\tilde{x}(i) + B\tilde{u}(i)\\
& ~~~~~~~~~~~\tilde{x}(0)  =  A\tilde{x}(N_r-1) + B\tilde{u}(N_r-1) \label{eq:stage1:term}\\
& ~~~~~~~~~~~(\tilde{x}(i), \tilde{u}(i) )\in \mathcal{Z}_i;~~i = 0, \ldots, N_{r}-1 \label{eq:stage1:feas}\\
& ~~~~~~~~~~~M^{def}=M^{shp}=0 \nonumber
\end{align}
\end{subequations}
where $T_1$ and $T_2$ are respectively positive definite and positive semi-definite
weight matrices of appropriate dimensions. The terminal constraint (\ref{eq:stage1:term}) ensures that batteries recover their initial SOC at the end of the control horizon.
\begin{remark}
  At \emph{stage-1}, a cost function is chosen that
penalizes the generation control input and the deviation of voltage from its nominal value. 
We define $T_2$ to be positive semi-definite since tracking the nominal voltage improves 
the power quality for loads, but the constraint (\ref{eq:stage1:feas}) is enough to ensure that voltage remains between 
operational bounds at every time step  $i \in \{0, \ldots, N_{r}-1\}$.
\end{remark}

\subsection{Stage 2: Model Predictive Controller}
In the second stage, a predictive controller is designed to minimize the overall cost of the system for time steps in 
$\{0,..,N-1\}$, with $N< N_r$. {Problem \emph{stage-1} is computed once, at the begining of the horizon, and we make use of the corresponding solution to ensure that the \emph{stage-2} problem remains recursively feasible under a receding horizon strategy.} Let $\lambda(t)$ denote the price of electricity at time t. We assume that $\lambda(t)$ is given as an input to the MPC and reflects demand peaks and congestion in the grid.
We propose the following MPC
problem (referred to as \emph{stage-2} in this paper):
\begin{equation*}
V_{2,k}(x,u) = \sum\limits_{i=0}^{N-1} \lambda(i+k) \big(\sum_{j=1}^{M^{shp}} u^{shp}_j(i+k)\big) + ||\nu-\nu_{nom}||^2_{T_3} \\.
\end{equation*}

\begin{subequations} \label{eqn:opt2}
 \begin{align}
   \operatornamewithlimits{min}\limits_{{x}, {u}}~~~~~&   V_{2,k}(x,u)\\
    \mbox{s.t}~~~~~&  {x}(i+1+k)  =  A{x}(i+k) + B{u}(i+k)\\
     &x(k)  = {x}_k \\
  &({x}(i+k), {u}(i+k)) \in \mathcal{Z}_{i+k};~~i = 0, \ldots, N-1\\
 &x(N+k) \in \mathbb{X}_{N+k} \label{eqn:opt2_help1}
  \end{align}
\end{subequations}

  In the MPC problem (\ref{eqn:opt2}), the contribution from fixed loads $p^l$, and deferrable loads $p^{def}$ is uncontrollable, therefore
   load shaping can be achieved only by controlling $u^{shp}$. The weight matrix $T_3$ is positive semi-definite
 and can be used to penalize large voltage deviations from $v_{nom}$.

\begin{remark}
  The solution from \emph{stage-1} is used to define the terminal set $\mathbb{X}_{N+k}$
  in the next Section. Thus we choose the horizon time $N$ such 
  that $N<N_r$.
\end{remark}

\subsection{Terminal Set} \label{sec:TerminalSet}
In this section, we detail the terminal set (\ref{eqn:opt2_help1}). This 
constraint must ensure that the system has enough flexibility to charge 
shapeable and deferrable loads before their plug-out time $k^{out}$, and keep 
voltage between the regular bounds after the control horizon N. {To ensure this, the terminal set is defined such that the battery banks have enough energy to meet the real power demand of additional (shapeable and deferrable) loads at the end of horizon. Moreover, the capacitors supply the reactive power to satisfy the network constraints under base load. Feasibility of the \emph{stage-1} problem thus ensures that all network constraints are satisfied even with the additional loads.}

The terminal constraint for the shapeable loads\textquotesingle ~ SOC $x_1=(e^{shp}_1,\ldots,{e}^{shp}_{M^{shp}})^T$ should guarantee that each load can be fully charged before their plug-out time. This is given by Eq. 
(\ref{eq:shapeableconstraint}), which ensures 
the recursive feasibility of the constraint  $e_{min} (k)\leq x(k) \leq e_{max}$, 
where $e_{min} (k^{out})=e^{shp}_{des}$. 

The terminal constraint for
the battery banks\textquotesingle {~}SOC $x_2=(e^{bat}_1\ldots,e^{bat}_n)^T$ guarantees that batteries have enough 
capacity at time $N$
to track the reference signal from \emph{stage-1} and supply the additional deferrable and shapeable loads.
We denote $(\hat{q}_i,~ \hat{e}_i^{bat})$ the optimal solution to the \emph{stage-1} problem in (\ref{eq:stage1}),
and $k^{out}_{max}$ the maximum plug-out time of all the deferrable and shapeable loads that are currently connected to the grid.

\begin{definition} 
\label{def:term_set}

We define the terminal set $\mathbb{X}_{N+k}$ at time $N+k$ as follows:
$\mathbb{X}_{N+k} :=$
\begin{align*}
&[x_1 (N\mbox{+}k), x_2(N\mbox{+}k) ] ^T=[e^{shp}(N+k),e^{bat}(N+k)]^T \\
  &\mbox{such that} ~~\forall i \in \{1,...,n\} ~~\forall j \in \{1,...,M^{shp}\} : \\
 &e^{shp}_{j}(N+k) \geq
 e^{shp}_{des,j}\mbox{ - }\mbox{\small{max}} (0,{\small(k^{out}_{j}\mbox{ - }(N\mbox{+}k) \small)}{\eta\Delta T}{c^{shp}_{\small{max},j}})\\
  &e^{shp}_{j}(N+k) \geq e^{shp}_{low,j}\\
  &e^{shp}_{j}(N+k)\leq e^{shp}_{des,j}\\
  &{e}_j^{bat}(N\mbox{+ }k)= \hat{e}_i^{bat}(N\mbox{+ }k)+\sum\limits_{
  l=N+k}^{k^{out}_{max}}\Delta T {p}_i^{def}(l)+\tilde{p}_i^{shp}(l)
\end{align*}
where 
\begin{eqnarray}
\tilde{p}^{shp}(l)&=& K^{shp} \tilde{c}^{shp}(l) \label{eqn:term_def1}\\
\tilde{c}^{shp}_j(l)&=&\frac{{e}^{shp}_{des,j}- {e}^{shp}_{j}(N+k) }{(k^{out}_j-(N+k))\eta \Delta T}1_{l< k^{out}_j}\label{eqn:term_def2}
\end{eqnarray}
and $k^{out}_j$ is defined in Eq. (\ref{eq:SOC_feas}).

\end{definition}

\begin{lemma}\label{theorem:terminal}
If the following conditions hold:
$\forall l \in [N+k, k^{out}_{max}]$, $\forall i \in {1,...,n}$:
\begin{align}
&\hat{p_i}^{bat}(l)- p_i^{def}(l)- 
\tilde{p}^{shp}_i(l)\geq {p_i}^{bat}_{min} \label{eq:cond1}\\
&\hat{e}_i^{bat}(l)+\sum\limits_{m=l}^{k^{out}_{max}}\Delta T {p}_i^{def}(m) \label{eq:cond2}\\
& ~~~+({e}^{shp}_{des,i}- {e}^{shp}_{i}(N+k) ) \frac{(k^{out}_i-l) }{\eta (k^{out}_i-(N+k)) }\leq {e}_{max,i}^{bat} \nonumber 
\end{align}
Then problem (\ref{eqn:opt2}) with terminal set $\mathbb{X}_{N+k}$ as defined in Definition \ref{def:term_set}, is recursively feasible, i.e. if the MPC optimization problem is feasible for $x(k)$, then it is also feasible for $x(k+1)$ defined in Eq. (\ref{eqn:sys1}).
\end{lemma}

\bigskip

\begin{proof}
This section gives a sketch of the proof, which is detailed in the Appendix. 
We define a feasible control sequence for \emph{stage-2} at time $k+N$, $( p^{shp}(k+N), {p}^{bat}(k+N), {q} (k+N))$, based on 
the solution of \emph{stage-1} at time $k+N$, $( \hat{p}^{shp}(k+N), \hat{p}^{bat}k+(N), \hat{q} (k+N))$:
\begin{eqnarray}
c_j^{shp}(N+k)&\mbox{=}&\frac{{e}^{shp}_{des,j}- {e}^{shp}_{j}(N+k) }{(k^{out}_j-(N+k))\eta \Delta T}1_{N+k< k^{out}_j} \label{eq:sol1} \\
  q_i(N+k)&\mbox{=}&\hat{q}_i(N+k) \label{eq:sol2}\\
p_i^{bat}(N+k)&\mbox{=}&[\hat{p_i}^{bat}\mbox{- }p_i^{def}\mbox{- 
}p^{shp}_i](N+k)\label{eq:sol3}
\end{eqnarray}
In practice, Eq (\ref{eqn:term_def1}) and (\ref{eqn:term_def2}) define a feasible control sequence after time $k+N$ where the shapeable power $\tilde{c}_j^{shp}$ at load j is constant until the plug-out time $k^{out}_j$.
Then, the appendix shows that Problem (\ref{eqn:opt2}) is recursively feasible assuming that Equations (\ref{eq:cond1}), (\ref{eq:cond2}) are true.
\end{proof}

\begin{remark}
Lemma 1 shows that under conditions (\ref{eq:cond1}), (\ref{eq:cond2}) the MPC problem is feasible at all times, if it is feasible for an initial state $x_0$. In next section we will show that the P\&P operation ensures that conditions (\ref{eq:cond1}), (\ref{eq:cond2}) are always satisfied, proving constraint satisfaction at all times. 
\end{remark}

\section{Plug-And-Play EV Charging} \label{secn:P&P}
In real distribution systems users can connect or disconnect their appliances randomly, including PEVs. 
This changes the overall load on the system and can affect bus voltages 
significantly. This section extends the MPC scheme to the case where the system dynamics 
in (\ref{eqn:sys1}) change due to loads joining or leaving the network by employing the concept of P\&P 
MPC in \cite{Melanie}. The introduction of P\&P capabilities poses two key 
challenges (\cite{Melanie}, \cite{Riverso}): (i) P\&P operations may produce infeasible operating conditions; (ii) the control law has to be redesigned for the modified dynamics. 
In the considered case, the problem is reduced to the first issue since the controller is computed centrally 
and the \emph{stage-2} MPC (c.f. Section \ref{secn:controller}) with modified dynamics directly produces the desired control law. In this section, we address the first challenge by means of a
preparation phase ensuring recursive feasibility and stability during P\&P operation. We first address the case of shapeable loads, then deferrable loads.

  \subsection{Shapeable loads} \label{sec:shp}
  As shown in Eq. (\ref{eqn:shp2}), we assume that a shapeable load can be plugged-in without drawing energy from the grid: $0\leq  c^{shp} \leq c^{shp}_{max}$. Therefore it is always optimal for a shapeable load to plug-in as soon as it requests it. That is, it can plug-in with $c^{shp}=0$ and wait for the system to allow strictly positive values $c^{shp}>0$. Thus, the output of the P\&P stage is to accept shapeable requests immediately. {Additionally, we assume that it is feasible to meet the user's requirements, i.e. fully charge the load before the maximum required time $k^{out}$, and satisfy equations (\ref{eq:cond1}), (\ref{eq:cond2}). For each new shapeable request, we check feasibility of the system by solving Eq (\ref{eqn:opt2}).} In practice, if a user makes an infeasible request, he would be asked to lower his/her requirements by allowing a later $k^{out}$ or a lower desired SOC $e^{shp}_{des}$.
  
  \subsection{Deferrable loads}\label{sec:def}
  The goal of the P\&P operation is to find a time to safely connect deferrable loads, and modify the response of shapeable loads and control devices to allow this connection as soon as possible. In this section we define a Mixed Integer Program (MIP) to find the minimum time to safely plug-in a deferrable load. After finding this time, we update the deferrable loads in the system and execute \emph{stage-2} with the new system. 
  
Deferrable loads do not impact the dynamics of the system, and only change the 
feasible set. An additional deferrable load at node j modifies the set $\mathcal{Z}_k$ via the equality:
 \begin{equation*}
P_{ij} =   p_{j}^{l} + p_{j}^{bat} +p_{j}^{def}+p_{j}^{shp}+ r_{ij}l_{ij} + \sum\limits_{k:(j,k) \in \mathcal{L}} P_{jk}
\end{equation*}
Moreover, it modifies the terminal set $\mathbb{X}_{N+k}$.
  In the following, we note $\overline{\mathcal{Z}}_k$ (respectively $\overline{\mathbb{X}}_{N+k}$) the feasible set (respectively terminal set) constraints that remain unchanged when a deferrable load plugs in.

Let's consider a P\&P request from a deferrable load at time $k$. The request can be postponed by $d_{max}<N$ maximum time steps. This creates $d_{max}$ possible load shapes.
For each possible time-delay $0 \leq d \leq d_{max}$ we note $p_{j}^{new,d}$ the corresponding shifted vector: 
\begin{equation}p_{j}^{new,d}=[\underbrace{ 0,...,0}_\text{size d}, p_{j}^{new,0}]\end{equation}. Thus we derive the new constraints when a deferrable load requests to plug-in at node $j$ and is delayed by $d$ time steps:

\begin{equation*}
P_{ij} =   p_{j}^{l} + p_{j}^{bat} +p_{j}^{def}+p_{j}^{shp}+ r_{ij}l_{ij} + \sum\limits_{m:(j,m) \in \mathcal{L}} P_{jm} + p_{j}^{new,d} 
\end{equation*}

\begin{align*} 
&{e}^{bat}(k+N)= \hat{e}^{bat}(k\mbox{+ }N)+\sum\limits_{m=k+N}^{k^{out}_{max}}\Delta T{p}^{def}(m) \\
&+\sum\limits_{m=k+N}^{k^{out}_{max}}\Delta Tp^{new,d}(m) + \frac{1}{\eta}K^{shp}\Big(e^{shp}_{des}-e^{shp}(N+k)\Big)
  \end{align*}

The solution $({x^*}, {{u}}^*,{z^*})$ 
of the following Mixed Integer Program (MIP) (\ref{eq:MILPdefferable})
gives the optimal transition time $d^*=\sum_{m=0}^{d_{max}} m z_m^* $.

\allowdisplaybreaks
\begin{subequations} \label{eq:MILPdefferable}
 \begin{align}
   \operatornamewithlimits{min}\limits_{{x}, {u}, z}~~~~~&   V_3(u,z)\\
    \mbox{s.t}~~~~~&  \mbox{\emph{System Dynamics:}} \nonumber \\
    &{x}(l+1)  =  A{x}(l) + B{u}(l)\\
     &{x}_0  =  x(0)\\
  &(x(l), u(l), P_{ij}(l), Q_{ij}(l), \nu_i(l), l_{ij}(l)) \in \overline{\mathcal{Z}_l}\nonumber\\
 &x(N+k) \in \overline{\mathbb{X}}_{N+k} \\
 & \mbox{\emph{Connection:}} \nonumber\\
 &z_m \in \{0,1\} ~~ \forall m\in \{0,1,...,d_{max}\}\\
   &\sum_{m=0}^{d_{max}} z_m=1  \\
   &\mbox{\emph{Power flow:}} \nonumber \\
 &P_{ij}(l) =   [p_{j}^{l}+ p_{j}^{bat} +p_{j}^{def}+p_{j}^{shp}](l)+ r_{ij}l_{ij}(l) \nonumber \\
 &~~~~~+ \sum\limits_{m:(j,m) \in \mathcal{L}} P_{jm}(l) + \sum_{d=0}^{d_{max}} z_d  p_{j}^{new,d}(l) \\
 & \mbox{\emph{Battery banks:}} \nonumber\\
&{e}^{bat}(N+k)= \hat{e}^{bat}(N+k)\\
&~~~~~~~ + \frac{1}{\eta}K^{shp}\Big(e^{shp}_{des}-e^{shp}(N+k)\Big)\nonumber\\
&~~~~~~~ +\sum\limits_{r=N+k}^{k^{out}_{max}} \Delta T [{p}^{def}+\sum_{m=0}^{d_{max}}z_m p^{new,m}](r)\nonumber\\
  &{p_j}^{bat}_{min} \leq \hat{p_j}^{bat}(s)\mbox{- }[p_j^{def} \mbox{+ } \sum_{m=0}^{d_{max}} z_m  p_{j}^{new,m} ](s)  \nonumber \\
 & ~~~~~~- \frac{{e}^{shp}_{des,j}- {e}^{shp}_{j}(N+k) }{(k^{out}_j\mbox{- }(N+k))\eta} \label{eq:PP_cond1}\\
  &{e}_{max,j}^{bat} \geq \sum\limits_{r=s}^{k^{out}_{max}} \Delta T [{p_j}^{def}\mbox{+ }\sum_{m=0}^{d_{max}}z_m p_j^{new,m}](r)  \label{eq:PP_cond2}\\
&~~~~~~ \mbox{+ } \frac{({e}^{shp}_{des,j}\mbox{- } {e}^{shp}_{j}(N+k) )(k^{out}_j\mbox{- }s) }{\eta (k^{out}_j\mbox{- }(N+k)) }\mbox{+ }  \hat{e}_j^{bat}(s) \nonumber\\
 \nonumber \\
 &~ s \in [N+k, k^{out}_j]\nonumber\\
 &l = k, \ldots, k+N-1\nonumber
  \end{align}
\end{subequations}

with
\begin{eqnarray}
  V_3({u},z)&=&\sum_{m=0}^{d_{max}} m z_m \label{eq:obj2}
\end{eqnarray}
Objective  (\ref{eq:obj2}) minimizes the transition delay and ensures that the problem remains feasible when the load plugs in. 

\begin{remark}
{Constraints (\ref{eq:PP_cond1}), (\ref{eq:PP_cond2}) correspond to the conditions in Lemma \ref{theorem:terminal} Eq. (\ref{eq:cond1}), (\ref{eq:cond2}) respectively.}
\end{remark}
We execute the request by (i) updating the system with the new load that plugs-in at time $d^*$, and (ii) going back to \emph{stage-2}. If $d^*>0$, then the control devices and shapeable loads update their signal during the transition phase $[N+k,N+k+d^*]$. The full controller is shown in Fig. \ref{fig:flow}.
 
\begin{figure}
\centering
\begin{tikzpicture}[scale=0.5]
\node[draw, fill=gray!35](P1) at (0,2) {Stage 1};
\node[draw, fill=gray!35](P2) at (0,-1) {Stage 2};
\node[draw, diamond, aspect=3](R) at (0,-4) {Plug-in request?};
\node[draw, diamond, aspect=3](N) at (4,-7) {No};
\node[draw, diamond, aspect=3](Y) at (-4,-7) {Yes};
\node[draw, fill=gray!35](PP) at (-4,-9) {P\&P MIP};
\draw (P1) [->, >=latex] to (P2);
\draw (P2)[->, >=latex] to (R);
\draw (R) [->, >=latex] to (Y);
\draw (Y) [->, >=latex] to (PP);
\draw (R) [->, >=latex] to (N);
\draw (N) [->, >=latex] to[bend right=120 ] (P2) ;
\draw (PP) [->, >=latex] to[bend left=120 ] (P2) ;
\end{tikzpicture}
\caption{Full controller flow: the solution at \emph{stage-1} is used to define the terminal set at \emph{stage-2}. When a new deferrable load requests to plug-in, the MIP determines the optimal plug-in time, the system is updated with the new load and the controller executes \emph{stage-2} on the new system.}
\label{fig:flow}
\end{figure}
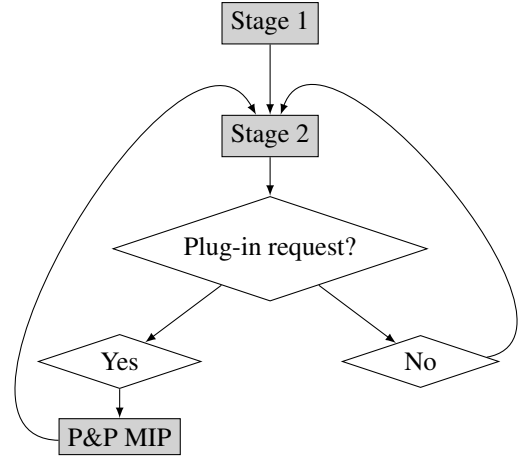

\bigskip

 \begin{theorem} \label{theorem:final}
The model predictive controller (\ref{eqn:opt2}) with network reconfigurations and transition times given by the MIP (\ref{eq:MILPdefferable}) is recursively feasible. For all initially feasible state $x_0$ and for all optimal sequences of control inputs, the controller optimization problem with P\&P network modifications (Fig. \ref{fig:flow}) remains feasible for all time.
 \end{theorem}

 \begin{proof} \label{proof:final}
 {Assume the problem is feasible at time $k$ and a request occurs at time $k$. The P\&P MIP (\ref{eq:MILPdefferable}) ensures that all constraints are satisfied during the transition time. Further, it provides that the conditions in Lemma \ref{theorem:terminal} are satisfied for the modified network. Hence the overall procedure maintains feasibility during transition and recursive feasibility is ensured after the modification.}
 \end{proof}

\section{Numerical results}\label{secn:example}
In this section we show simulation results on a 55 bus Southern California Edison distribution network (see Fig. \ref{fig:feeder}). This network was previously studied in \cite{Farivar2012}. 
We model seven additional storage devices at nodes 2, 8, 10, 14, 21, 30, 41 (represented in green in Fig. \ref{fig:feeder}). We assume that the price of electricity is given and reflects the requirements of the system operator.
In this case study we choose the time step $\Delta t= 0.5$h, the \emph{stage-2} MPC time horizon $\frac{N}{\Delta T}=$5h and the \emph{stage-1} time horizon $\frac{N_r}{\Delta T}=$48h. In this section, we illustrate the controller signal for a period of 30h in order to show daytime and nigh-time load schedules.

\begin{figure}[h]
	\centering
    \includegraphics[trim = 11mm 0mm 1mm 0mm, clip, width =0.5 \textwidth]{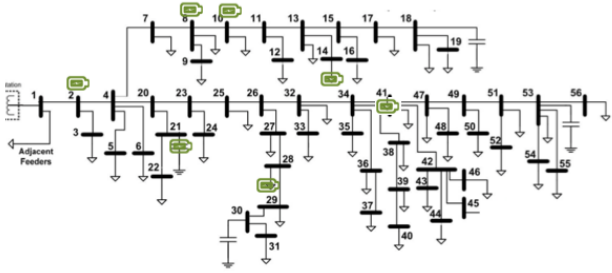}
        \caption{55 bus feeder. Additional battery banks are indicated in green (capacitors are not represented here).}
	\label{fig:feeder}
\end{figure}

\subsection{Load scheduling}

\begin{figure}[tb]
	\centering
    \includegraphics[trim = 21mm 2mm 26mm 10mm, clip, width = 0.5\textwidth]{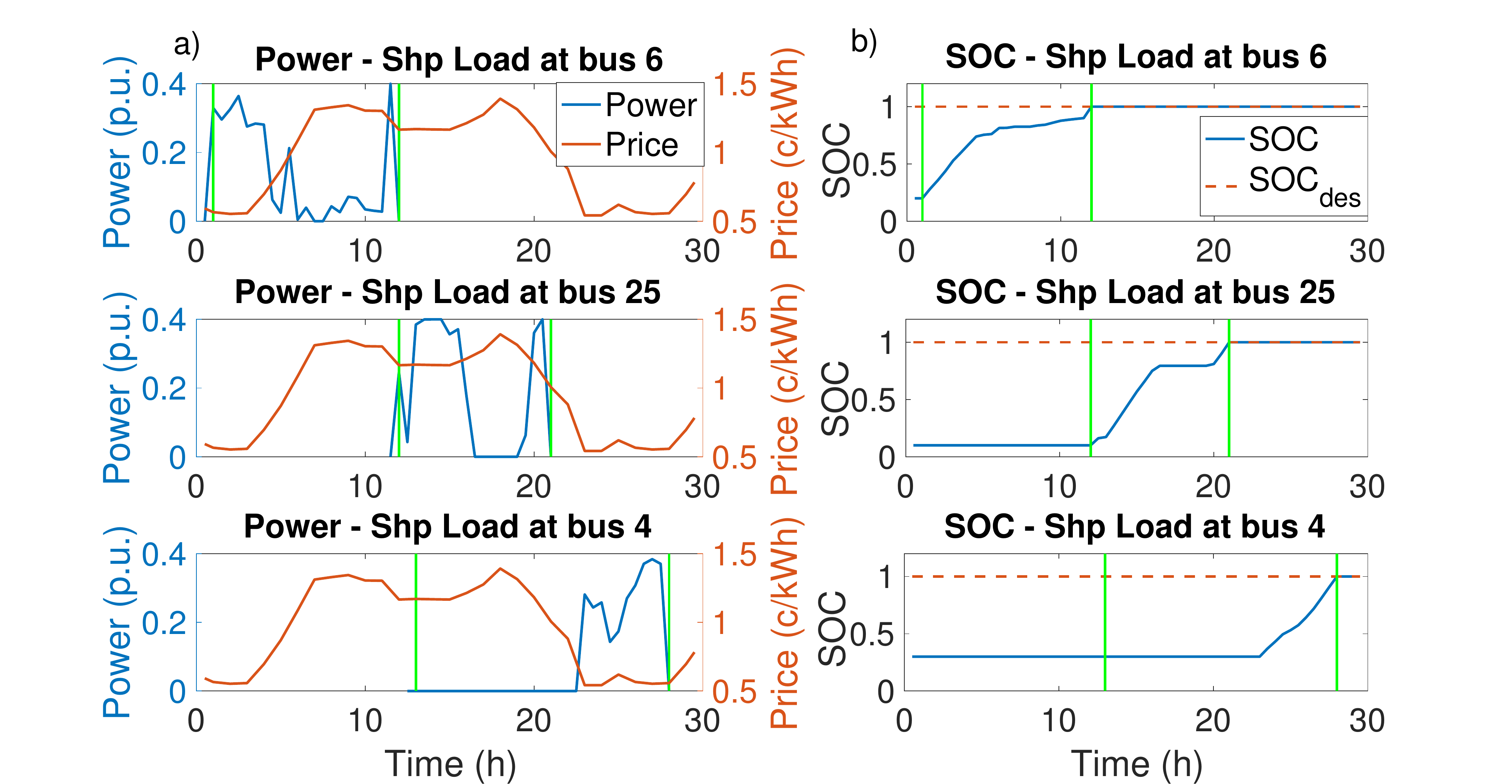}
	 \caption{Power and SOC of three different Shapeable Loads in the network. Green vertical lines show when the load requests to plug in and plug out.}
	     \label{fig:shp_plot}
\end{figure} 

\begin{figure*}[tb]
	\centering
    \includegraphics[trim = 11mm 0mm 1mm 0mm, clip, width =\textwidth]{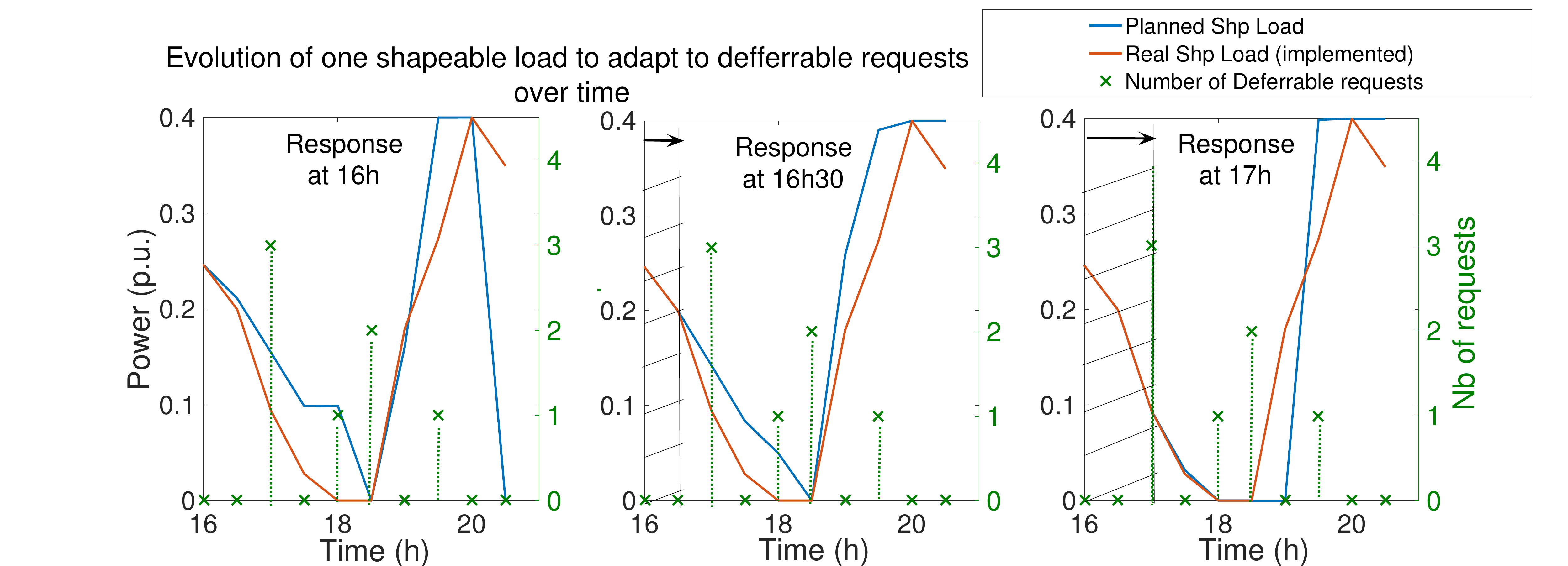}
        \caption{Evolution of one shapeable load to adapt to plug-in requests from deferrable loads. Three requests occur at 17h, forcing the shapeable load to lower its power and adapt to the new system.}
	\label{fig:shp_evolv}
\end{figure*} 

\begin{table}[b]
\caption{Description of Shapeable Loads in the System}
\label{tbl:shp}
\begin{center}
    \begin{tabular}{| r | r| r| }
    \hline
  \multicolumn{3}{|c|}{Shapeable Loads} \\
    \hline
    Time (h) & Nb of Requests & Bus number\\ \hline
    1& 2 & 6, 9  \\ \hline
    8& 2 & 5, 19  \\ \hline
     11 & 1 & 15  \\ \hline
    12 & 1 & 25  \\ \hline
    13 & 2 & 4, 31  \\ \hline
    15 & 2 & 19, 19 \\ \hline
     16 & 2 & 15, 15  \\ \hline
    16.5 & 2 & 25, 15  \\ \hline
    \end{tabular}
\end{center}
\end{table}

\begin{table}
\caption{Description of Defferable Loads in the System}
\label{tbl:def}
\begin{center}
\begin{adjustbox}{max width=0.5\textwidth}
    \begin{tabular}{| r | p{1.3cm}| p{1.5 cm}| p{1.5 cm}|p{1.5 cm}|}
    \hline
  \multicolumn{5}{|c|}{Deferrable Loads} \\
    \hline
    Time (h) & Nb of Requests &  Nb of Deferred Requests & Bus number& Plug-in Time\\ \hline
    4& 1 & 0& 8 & 4  \\ \hline
    6& 1 & 0&33 & 6  \\ \hline
     10 & 5 & 0&4, 5, 5, 16, 17 &10, 10, 10, 10, 10   \\ \hline
    11 & 1 &1& 28 &11.5  \\ \hline
    12 & 2  & 1&19, 38 & 14.5, 12  \\ \hline
    17 & 3  & 0&8, 20, 22 & 17, 17, 17\\ \hline
    18 & 1 & 0&12 & 18\\ \hline
     18.5 & 2&0 & 5, 22 & 18.5, 18.5  \\ \hline
    19.5 & 1 &0& 12 & 19.5  \\ \hline
    \end{tabular}
    \end{adjustbox}
\end{center}
\end{table}

Tables \ref{tbl:shp} and \ref{tbl:def} show plug-in requests from shapeable loads and deferrable loads in the network. Table \ref{tbl:shp} shows the 14 shapeable requests, their request time and corresponding bus number. As mentioned in Section \ref{sec:shp}, shapeable loads plug-in as soon as they requests it, can be zero-power during a certain time and fully charge before their desired plug-out time $k^{out}$.
Figure \ref{fig:shp_plot} shows three shapeable loads and how they charge as a function of the electricity price. The vertical green lines show the connected period: the first green line is the request time, the second green line is $k^{out}$, i.e. the time between the two lines is the only time when the load is plugged-in. Figure \ref{fig:shp_plot}a shows that loads draw power only when they are plugged-in and Fig. \ref{fig:shp_plot}b shows that they reach their desired SOC before $k^{out}$. In these three examples, the loads tend to charge when the price of electricity is cheaper. In particular, the load at bus 4 avoids the evening peak time (between 5pm and 9pm), and charges during the night time (between 10pm and 6am).

 Shapeable loads have the flexibility to adapt their power signal to the conditions and constraints of the network. In particular, when a deferrable load requests to plug-in, shapeable loads can adapt their response to the new constraints to allow the new deferrable load to plug-in, without violating the network constraints. Figure \ref{fig:shp_evolv} shows the response of one shapeable load at three consecutive time steps: 16h, 16h30 and 17h. The response evolves across time and changes when deferrable loads request to plug-in.  In particular the power ramps down at 17h when three new requests occur. This allows the network to adapt to these new loads, and the three requests are all immediately accepted. Table \ref{tbl:def} shows that only two deferrable loads need to be deferred in this case: one load at time 11h and bus 28 and one load at time 11h and bus 19. Figure \ref{fig:total_load_evolv} shows that one deferrablein the evening load is delayed at 11h. During the transition phase (11h-11h30), shapeable loads adapt their signal to enable safe connection of the deferrable load at 11h30.

\begin{figure*}[tb]
   \begin{minipage}[c]{0.5\linewidth}
      \includegraphics[page=3,trim = 10mm 0mm 0mm 0mm, clip, width = \textwidth]{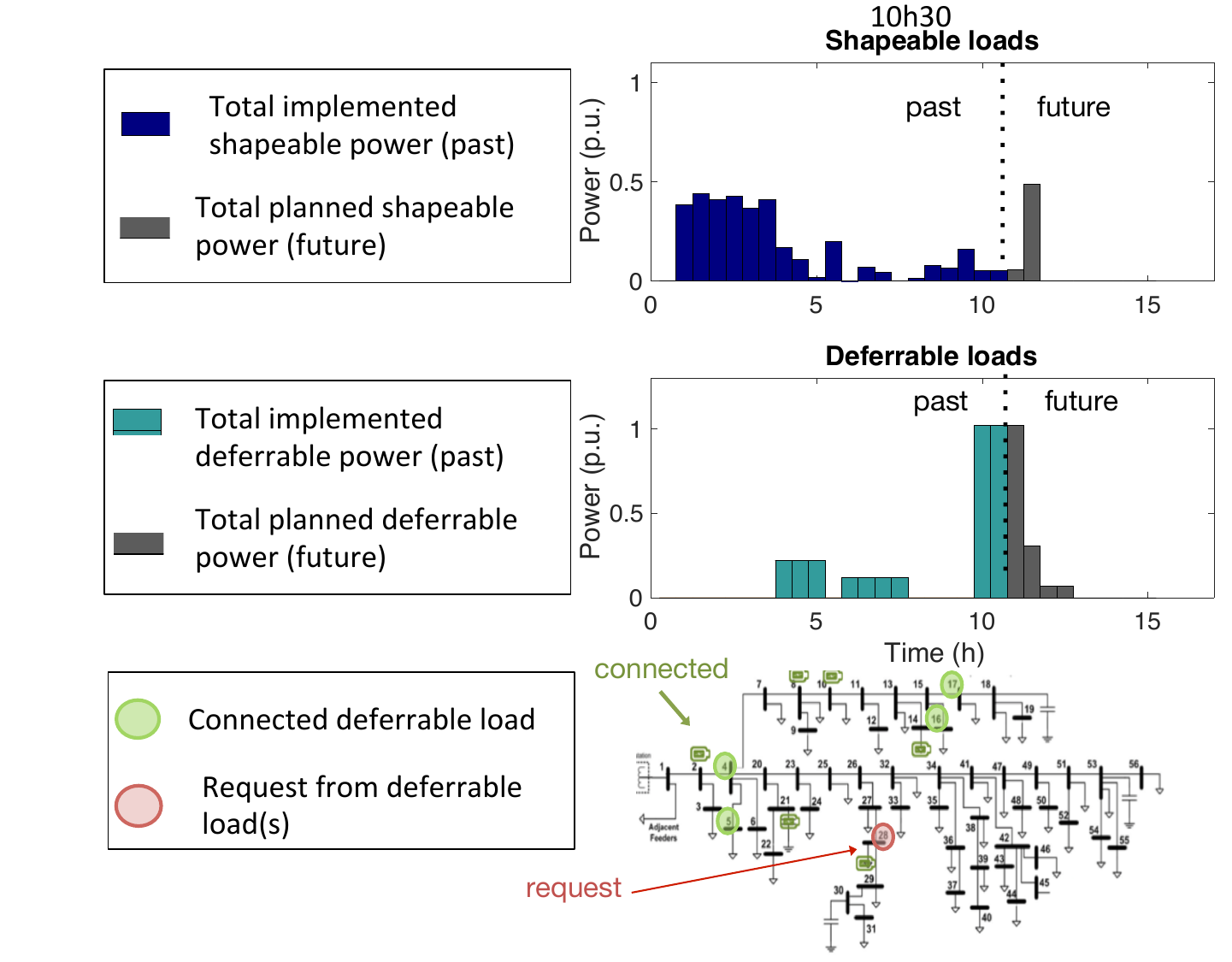}
   \end{minipage} \hfill
   \begin{minipage}[c]{0.5\linewidth}
      \includegraphics[page=5,trim = 7mm 0mm 0mm 0mm, clip, width =\textwidth]{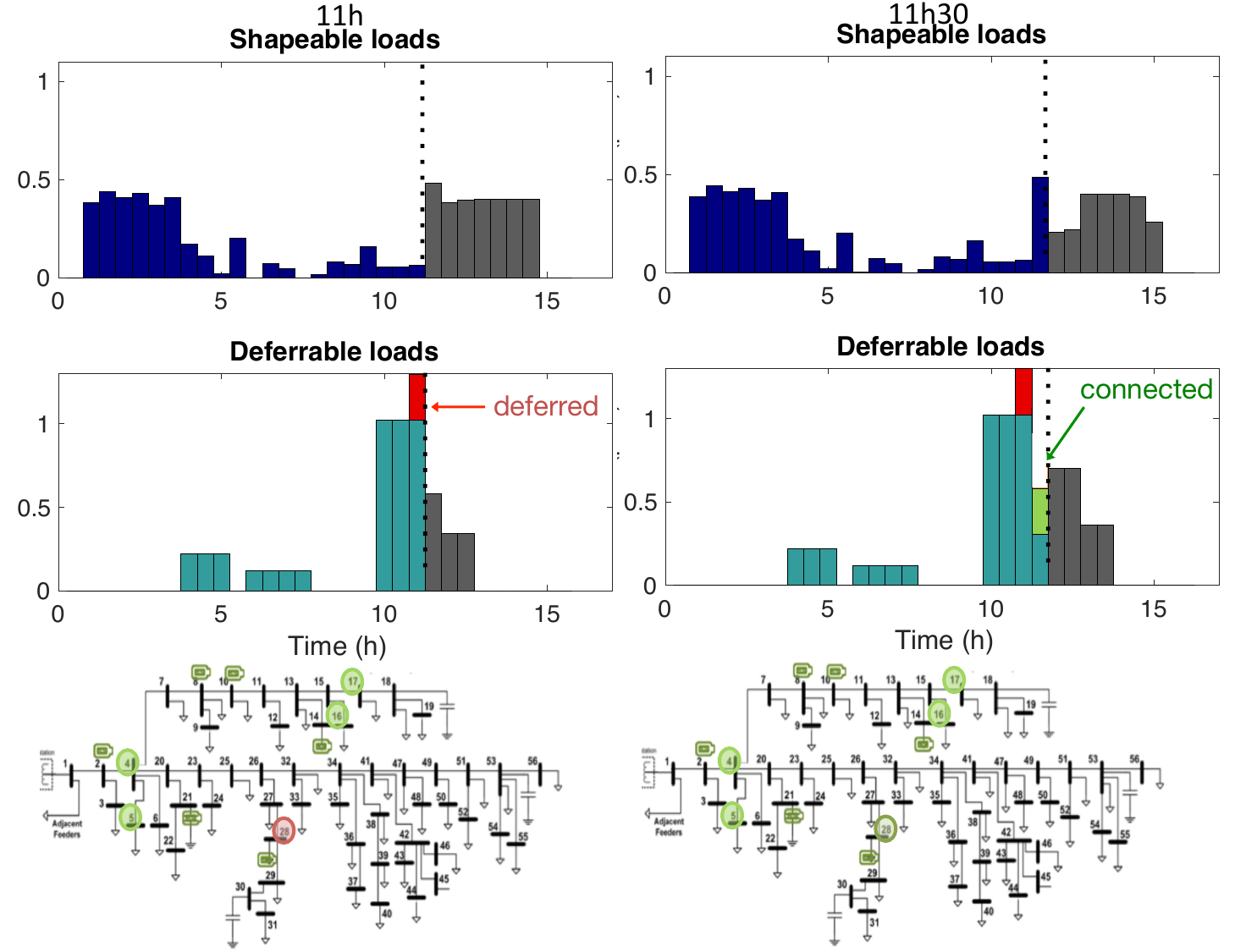}
   \end{minipage}
    \caption{Evolution of loads: one deferrable load requests to plug-in at 11h
 and is delayed to connect at 11h30. During the transition phase (11h-11h30), shapeable loads adapt their signal to enable safe connection of the deferrable load.}
      \label{fig:total_load_evolv}
\end{figure*}
\subsection{Network constraints}
In this section, we illustrate the network constraints, namely voltage constraints and battery banks constraints. Figure \ref{fig:voltage} shows the voltage at each bus and time step. It shows that voltage remains between the bounds 0.95 and 1 at all times.
\begin{figure}[tb]
	\centering
    \includegraphics[trim = 11mm 0mm 1mm 0mm, clip, width = 0.5\textwidth]{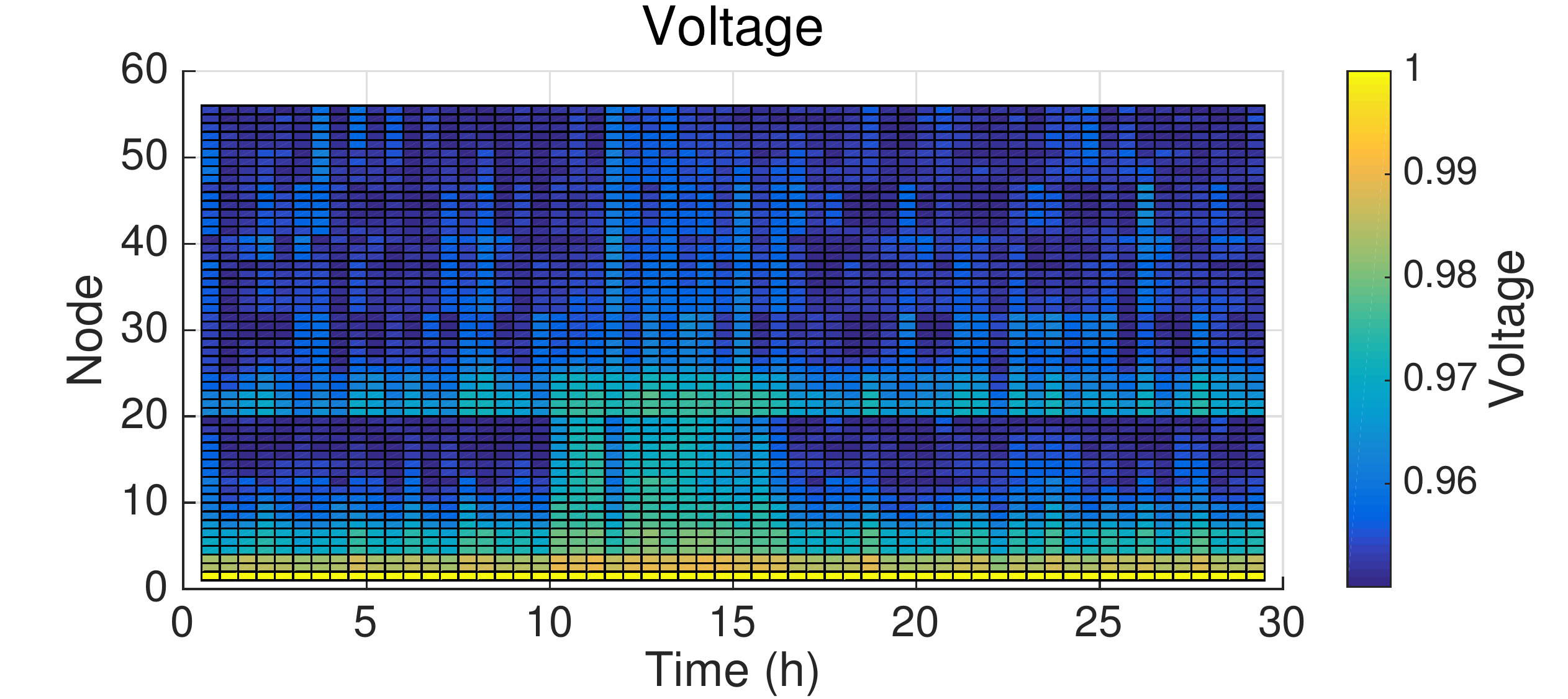}
	\caption{Voltage at each mode of the network}
     \label{fig:voltage}
\end{figure} 

Figure \ref{fig:batt} shows the power and SOC at the seven battery banks and Fig. \ref{fig:real_full} shows the aggregated real power over time. Figure \ref{fig:batt}a shows that batteries tend to highly discharge, i.e. have high negative power, around 10h, 15h and 20h. Figure~\ref{fig:real_full} shows that these are times when the network is highly loaded, i.e. a lot of shapeable loads and deferrable loads are connected and fixed loads are high. Storage devices are used to supply additional power in case of load peaks. Note that we impose the minimum SOC, $e^{bat}_{low}= 0.12$ however the SOC never goes below 0.3. This limit is due to the terminal constraints (\ref{eq:stage1:term}) in \emph{stage-1} of the problem: the initial SOC has to be recoverable at the final time $N_h=48h$.
\begin{figure}[tb]
	\centering
    \includegraphics[trim = 11mm 0mm 1mm 0mm, clip, width = 0.5\textwidth]{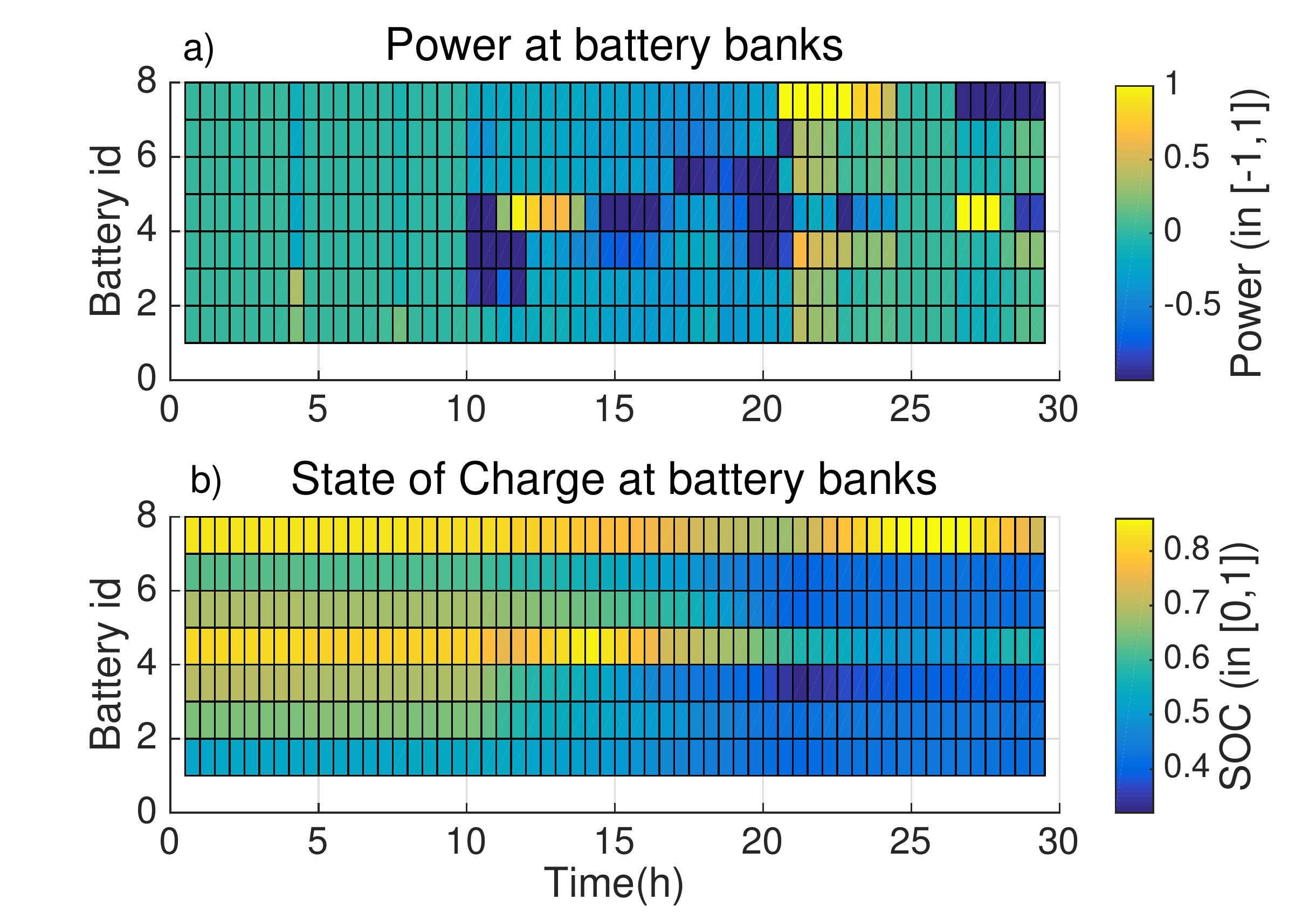}
    \caption{a) Real power and b) State of Charge at battery banks. Values are normalized }
     \label{fig:batt}
\end{figure} 

\begin{figure}[tb]
	\centering
    \includegraphics[trim = 29mm 0mm 6mm 0mm, clip, width =0.5 \textwidth]{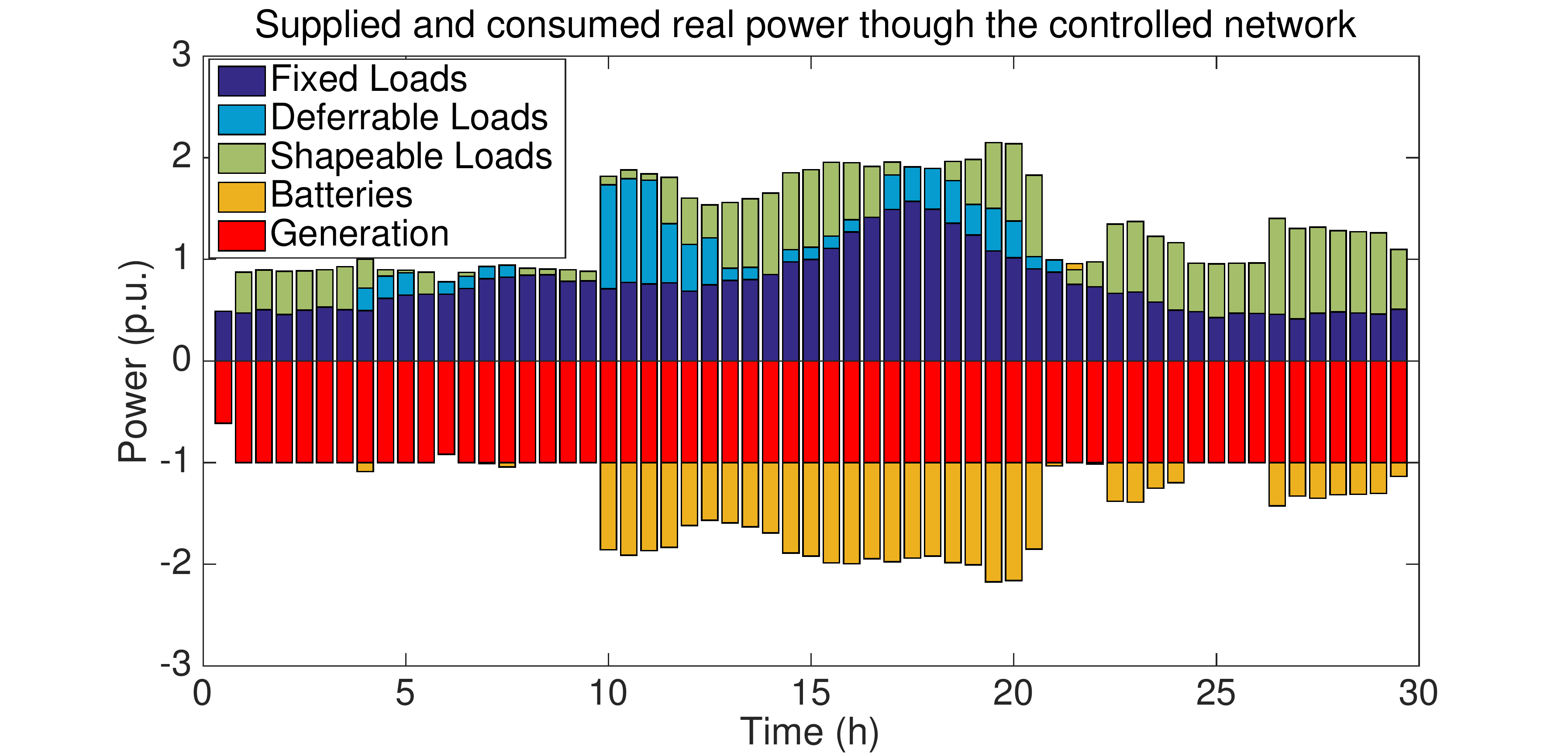}
        \caption{Real power across the different devices of the networks}
	\label{fig:real_full}
\end{figure}

\subsection{Peak reduction impact}
In this section we illustrate peak reduction impact of the controller. Figure \ref{fig:real_pow} shows the aggregate load in the network in three cases: a) in the uncontrolled case, b) when the controller is applied to the network with deferrable and shapeable loads and c) when the controller is applied to the network without deferrable loads. In the uncontrolled case, every load plugs in as soon as it requests it.  
\begin{figure}[tb]
	\centering
	\includegraphics[trim = 8mm 0mm 10mm 0mm, clip, width = 0.45\textwidth]{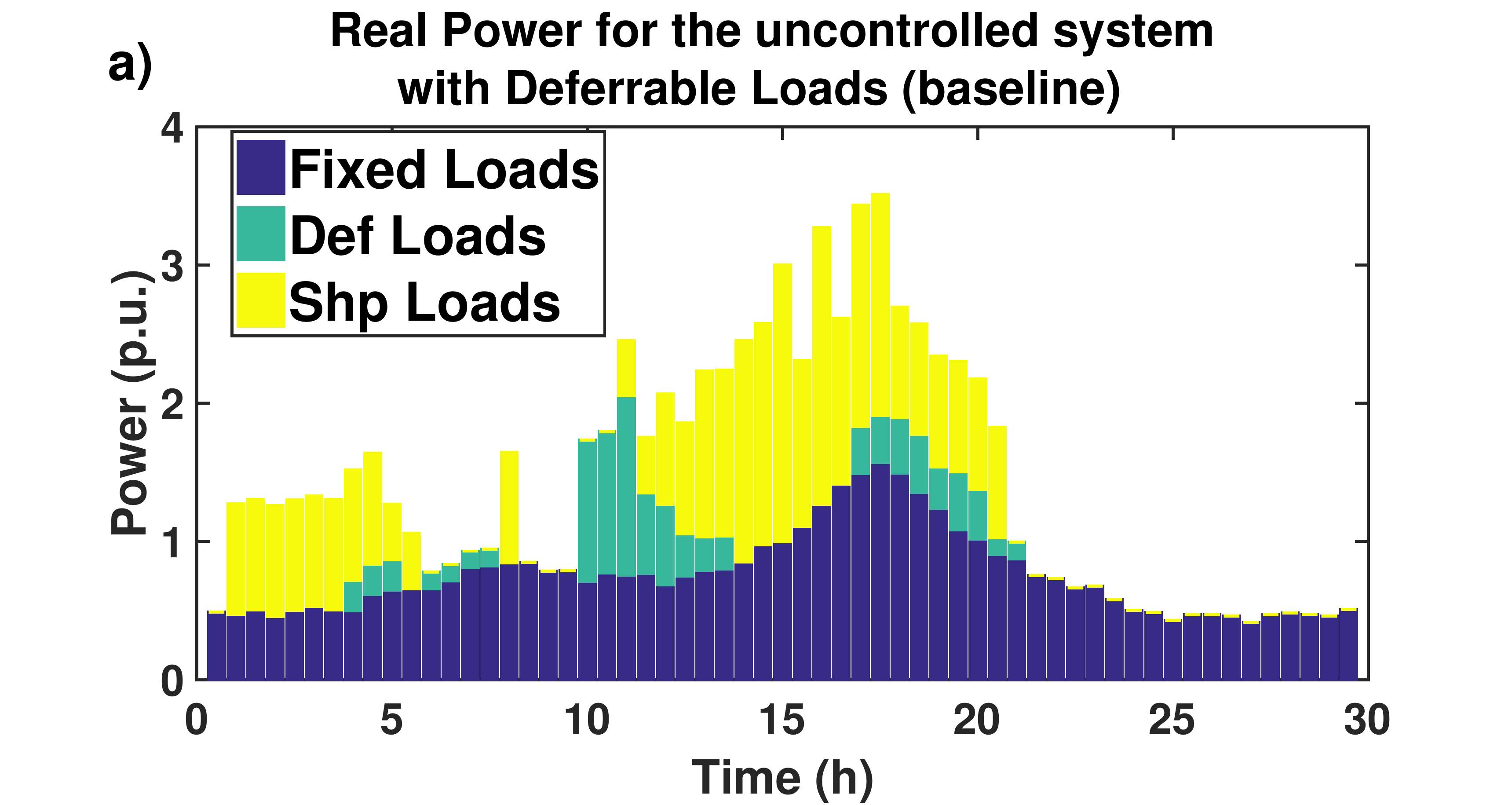}
	\centering
    \includegraphics[trim = 8mm 0mm 10mm 0mm, clip, width = 0.45\textwidth]{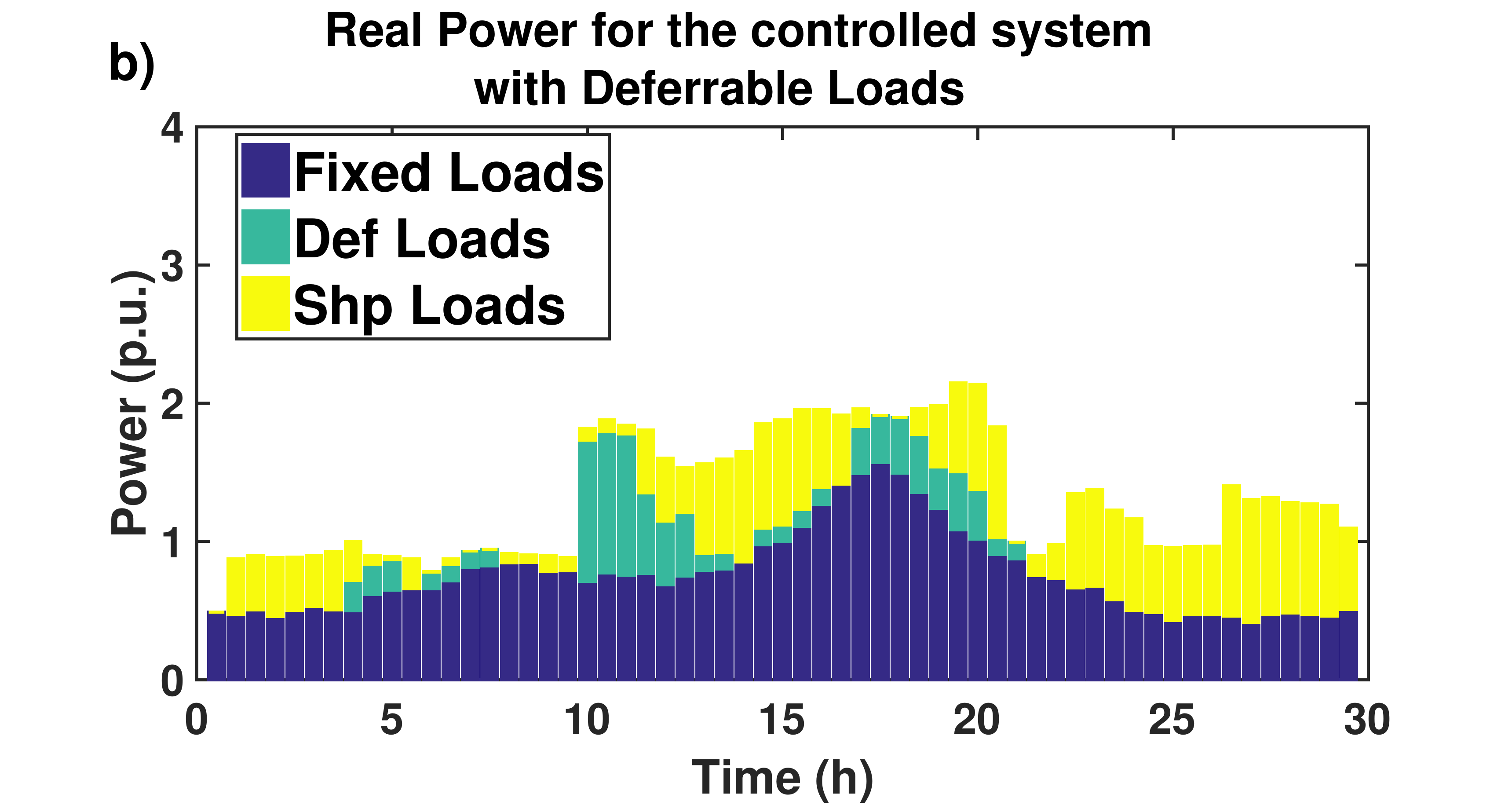}
    \centering
        \includegraphics[trim = 8mm 0mm 10mm 0mm, clip, width = 0.45\textwidth]{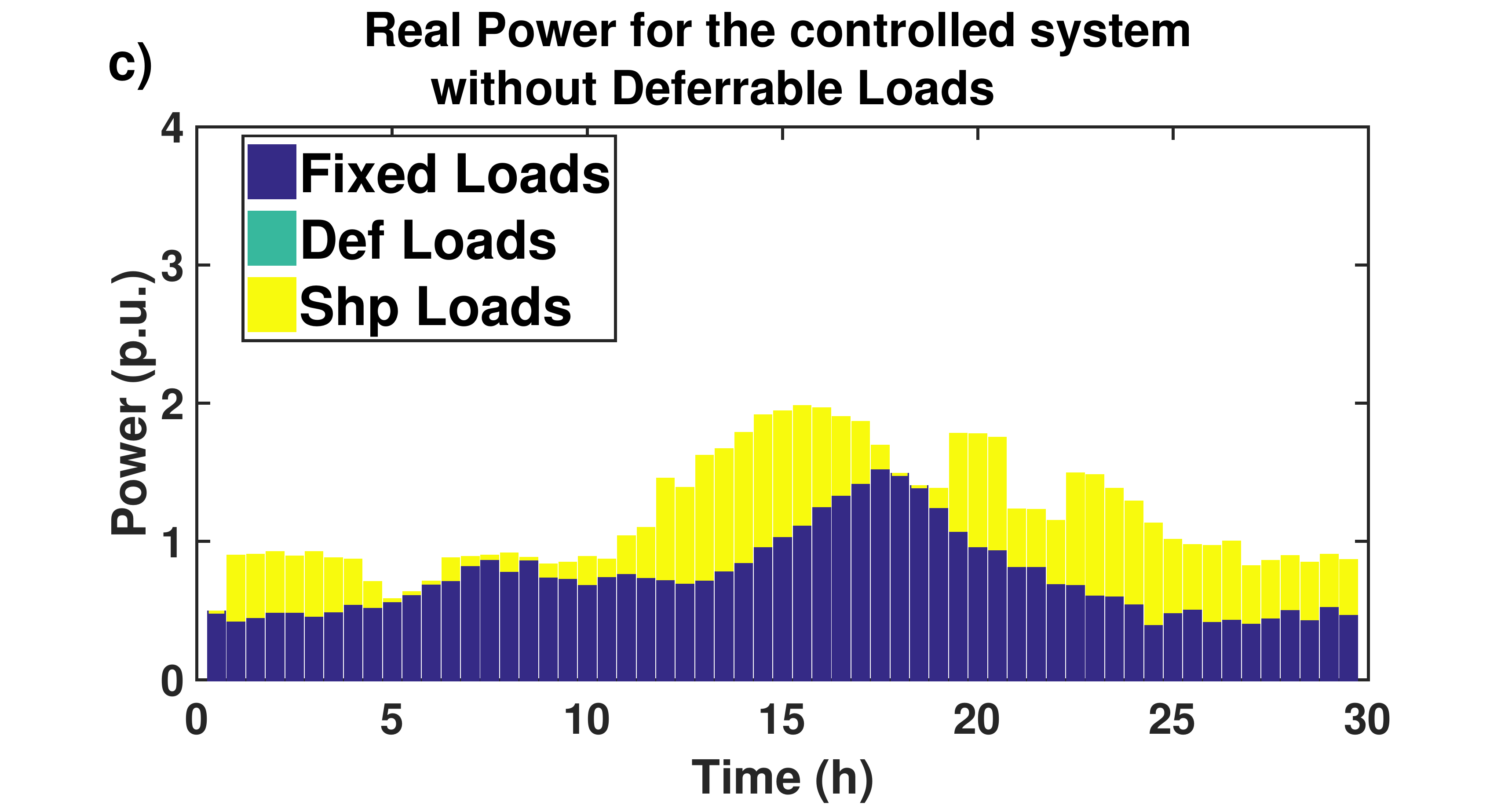}
    \caption{Cumulative real power in the network for a) the controlled system without deferrable loads, b) the controlled system with deferrable loads, c) the uncontrolled system with deferrable loads}
    \label{fig:real_pow}
\end{figure} 

The total peak in the uncontrolled case (Fig. \ref{fig:real_pow}a) is 3.5 p.u. whereas the total peak in the controlled case (Fig \ref{fig:real_pow}b) is 2 p.u, providing 40$\%$ reduction of the daily load peak. In the uncontrolled case, a lot of additional loads plug-in during the peak time (between 3pm and 9pm) and immediately charge. On the contrary, in the controlled case, shapeable loads are delayed to the night time, which results in a smoother load curve.

The difference between Fig. \ref{fig:real_pow}b and Fig. \ref{fig:real_pow}c illustrates how shapeable loads' schedules change when deferrable loads are connected to the network. Fig \ref{fig:real_pow}b shows that shapeable loads adapt their load profile to enable connection of deferrable loads: in Fig \ref{fig:real_pow}b shapeable power tends to be delayed to the night time, in order to allow connection of deferrable loads in the evening (6pm-9pm).

\section{Conclusion} \label{secn:conclusion}
In this paper, we have proposed a predictive controller that is 
capable of handling P\&P requests of flexible and deferrable loads. 
First, an MPC approach for minimizing the global cost of the system was used to aggregate flexible loads and provide load shaping objectives under distribution grid constraints. Second, we have proposed a MIP that safely connects loads and minimizes waiting times. We proved that our algorithm achieves recursive feasibility, by appropriately defining the connection conditions and the terminal constraint set. The performance of the proposed method was demonstrated for the control of a radial distribution system with 55 buses.
%



\section*{APPENDIX} \label{appendix}
Assume that the problem is feasible at time $k$, with the optimal control sequence
$U^*(k)=\left \{ u^*(k),u^*({k+1}),...,u^*({k+N-1})\right \}$ and the predicted state trajectory $X^*(k)=\left \{ x^*({k+1}),...,x^*({k+N-1}),x^*({k+N})\right \}$. Then, at time $k+1$, we show that the control sequence $U^*(k+1)=\left \{ u^*({k+1}),u^*({k+2}),...,u^*({k+N-1}), v({k+N})\right \}$ is feasible where $v({k+N})$ is defined by Equations (\ref{eq:sol1}), (\ref{eq:sol2}), (\ref{eq:sol3}). The state trajectory at time $k+1$ is $X^*(k+1)=\left \{ x^*({k+2}),...,x^*({k+N}),x({k+N+1})\right \}$, where $x({k+N+1})$ is derived in the next sections.   
With the notations in Section \ref{sec:DynSys}:
\begin{eqnarray*}
x^*({k+N}) &=& \begin{bmatrix} x_1(k+N), x_2(k+N) \end{bmatrix} ^T\\
&=& \begin{bmatrix} (e^{shp}_1,\ldots,{e}^{shp}_{M^{shp}}) ,(e^{bat}_1\ldots,e^{bat}_n)\end{bmatrix}^T(N) \\
x({k\mbox{+}N\mbox{+}1}) &=& \begin{bmatrix} x_1(k\mbox{+}N\mbox{+}1), x_2(k\mbox{+}N\mbox{+}1) \end{bmatrix} ^T\\
 v({k+N}) &=& \begin{bmatrix} q^{g}(k+N) ~ u^{shp}(k+N) ~ p^{bat}(k+N) \end{bmatrix}_{}^T 
\end{eqnarray*}
And  $v({k+N})$ is defined by the following control values:
\begin{align*}
c_j^{shp}(k\mbox{+}N)&\mbox{=}\frac{{e}^{shp}_{des,j}- {e}^{shp}_{j}(k+N) }{(k^{out}_j-(k+N))\eta \Delta T}1_{(k\mbox{+}N)< k^{out}_j} \\
&~~~~j \in \{1,...,M^{shp}\} \\
  q_i(k\mbox{+}N)&\mbox{=}\hat{q}_i(k\mbox{+}N) ~~~i \in \{1,...,n\}\\
p_i^{bat}(k\mbox{+}N)&\mbox{=}[\hat{p_i}^{bat}(k\mbox{+}N)\mbox{- }p_i^{def}(k\mbox{+}N)\mbox{- 
}p^{shp}_i(k\mbox{+}N)] 
\end{align*}
In the remainder of this section, we derive $x({k+N+1})$ and prove that the solution is feasible.

\subsection{Shapeable loads state of charge}
Let consider $j \in \{1,...,M^{shp}\}$:
\begin{align*}
  e^{shp}_j(k+N+1)&=e^{shp}_j(k+N)+\eta \Delta T c^{shp}_j(k+N)\\
  &= e_j^{shp}(k+N)\\
  &~+\eta \Delta T \frac{{e}^{shp}_{des,j}- {e}^{shp}_{j}(k+N) }{(k^{out}_j-(k+N))\eta \Delta T}1_{(k+N)\leq k^{out}_j} 
\end{align*}
At time $k+N:$ 
\begin{align*}
 &e^{shp}_j(k+N)\geq
 e^{shp}_{des,j}\mbox{ - }\mbox{\small{max}} (0,{\eta\Delta T}{\small{(k^{out}_j\mbox{ - }(k+N)}}{c^{shp}_{\small{max},j}})\\
   &e^{shp}_j(k+N)\geq e^{shp}_{min,j}   
\end{align*}
  Thus, we obtain the following condition at time $k+N+1$, which ensures
  recursive feasibility of the terminal constraint for shapeable loads $e^{shp}_j$:
  
  \begin{align}
 e^{shp}_j(k+N+1)&\geq e^{shp}_j(k+N)\geq e^{shp}_{min,j} 
  \end{align}
  
  If $N< k^{out}_j$
  
  \begin{align}
&e^{shp}_{des,j}- e^{shp}_j(k+N+1)\nonumber \\  &= e^{shp}_{des,j}-(e_j^{shp}(k+N)+\frac{{e}^{shp}_{des,j}- {e}^{shp}_{j}(k+N) }{(k^{out}_j-(k+N))}  \nonumber \\
 &=\Big(e^{shp}_{des,j}-e_j^{shp}(k+N)\Big)(1- \frac{1}{k^{out}_j-k-N})\nonumber\\
 &\leq \eta \Delta T (k^{out}-k-N) c^{shp}_{max,j}\frac{k^{out}_j-(k+N+1)}{k^{out}-N}\nonumber \\
  &\leq \eta \Delta T \Big(k^{out}-(k+N+1)\Big) c^{shp}_{max,j} 
   \end{align}
  Moreover:
  \begin{equation*}
  e^{shp}_{des,j}- e^{shp}_j(k+N+1) \geq 0 \mbox{  since $k^{out}_j-(k+N) \geq 1$}
  \end{equation*}

  If $N\geq k^{out}$
  \begin{align}
  e^{shp}_j(k+N+1)= e^{shp}_j(k+N)&\geq e^{shp}_{des,j}\\
  & \leq e^{shp}_{des,j}
  \end{align}

  \subsection{Network constraints}
  Consider a node $i \in \{1,...,n\}$.
Equations (\ref{eq:sol2}), (\ref{eq:sol3}) give:\\
$\begin{cases}
  q_i(k+N)=\hat{q}_i(k+N)\\
  p_i^{bat}(k+N)=[\hat{p_i}^{bat}\mbox{- }p_i^{def}\mbox{- }p^{shp}_i](k+N)
\end{cases}{~}\\$
Equation (\ref{eqn:orig_PF1}) gives:
\begin{eqnarray*}
  P_{ij}(k+N) &= &  [p_{j}^{l}+ \hat{p_j}^{bat}+ r_{ij}l_{ij}](k+N) \\&+&
  \sum\limits_{m:(j,m) \in \mathcal{L}} P_{jm}(k+N)
\end{eqnarray*}
Equation (\ref{eqn:orig_PF2}), (\ref{eqn:orig_PF3}) and (\ref{eq:convex}) give:
\begin{eqnarray*}
Q_{ij} (k+N)  &=  &[q_{j}^{l} - \hat{q}_{j}^{g} + x_{ij}l_{ij}](k+N) \\&+& \sum\limits_{k:(j,m) \in \mathcal{L}} 
Q_{jm}(k+N)\\
\nu_{j}(k+N)&  = &\nu_{i}(k+N) + (r_{ij}^2 + x_{ij}^2)l_{ij} \\&-& 2(r_{ij}P_{ij}(k+N) + 
x_{ij}Q_{ij})\\
l_{ij}(k+N)  &\geq& \frac{P_{ij}^2(k+N) + Q_{ij}^2(k+N)} {\nu_{i}(k+N)}
\end{eqnarray*}

This is the system of power flow equations for $(p,q)=(\hat{p},\hat{q})$. Thus it 
is feasible and the voltage bounds are satisfied:
\begin{equation*}
  \nu_{min}\leq\nu(k+N)\leq\nu_{max}
\end{equation*}

  \subsection{Battery banks}
  The power constraint and terminal constraint for the SOC of battery banks must be 
  satisfied. By induction, we show that $\forall l \in [k+N, k^{out}_{max}$]:
  \begin{equation}
   {e}_i^{bat}(l)=\hat{e}_i^{bat}(l)+\sum\limits_{m=l}^{k^{out}_{max}}\Delta T [{p}_i^{def}+\tilde{p}_i^{shp}](m) \label{induction}
  \end{equation}
  
  By definition of the terminal set, this is true at time $k+N$. Now, suppose it is true at time $l \in [k+N, {k^{out}_{max}}$], then:
  \allowdisplaybreaks
  \begin{align*}
    e_i^{bat}(l+1)&=e_i^{bat}(l) + \Delta T  p_i^{bat}(l) \\
    &= e_i^{bat}(l) + \Delta T [\hat{p_i}^{bat}(l)\mbox{- }p_i^{def}(l)\mbox{- }\tilde{p}_i^{shp}(l)]\\
    &=\hat{e}_i^{bat}(l)+\sum\limits_{m=l}^{k^{out}_{max}}\Delta T [{p}_i^{def}(m)+\tilde{p}_i^{shp}(m)]\\
    &~+ \Delta T [\hat{p_i}^{bat}(l)\mbox{- }p^{def}(l)\mbox{- }\tilde{p}^{shp}(l)]\\
    &=\hat{e}_i^{bat}(l+1)+\sum\limits_{m=l+1}^{k^{out}_{max}}\Delta T [{p}_i^{def}(m)+\tilde{p}_i^{shp}(m)]
  \end{align*}
  This is Eq. (\ref{induction}) at time $l+1$, proving that (\ref{induction}) holds by induction. Now, Eq. (\ref{induction}) at time $k+N+1$ gives:
  \begin{eqnarray*}
   {e}_i^{bat}(k+N+1)&=&\hat{e}_i^{bat}(k+N+1)\\&&+\sum\limits_{l=k+N+1}^{k^{out}_{max}}\Delta T [{p}_i^{def}(l)+\tilde{p}_i^{shp}(l)] \label{eqn:bat_N+1}
  \end{eqnarray*}

Thus ${e}_i^{bat}(k+N+1)\geq\hat{e}_i^{bat}(k+N+1)\geq {e}_{i,min}^{bat}$ and Assumption (\ref{eq:cond1}) gives ${e}_i^{bat}(k+N+1)\leq{e}_{i,max}^{bat}$. 
Moreover,
\begin{eqnarray*}
p_i^{bat}(k+N) &\mbox{=}&[\hat{p_i}^{bat}(k\mbox{+ }N)\mbox{- }p_i^{def}(k\mbox{+ }N)\mbox{- 
}p^{shp}_i(k\mbox{+ }N)]\\ &\leq& p_i^{bat}(k+N) \leq p_{i,max}^{bat}
\end{eqnarray*}
and condition (\ref{eq:cond2}) gives $p_i^{bat}(k+N)\geq p_{i,min}^{bat}$.\\
This concludes the proof of recursive feasibility.

\bibliographystyle{ieeetr}
\bibliography{EV_Charging}

\begin{thebibliography}{10}

\bibitem{wang2011survey}
W.~Wang, Y.~Xu, and M.~Khanna, ``A survey on the communication architectures in
  smart grid,'' {\em Computer Networks}, vol.~55, no.~15, pp.~3604--3629, 2011.

\bibitem{fan2013smart}
Z.~Fan, P.~Kulkarni, S.~Gormus, C.~Efthymiou, G.~Kalogridis, M.~Sooriyabandara,
  Z.~Zhu, S.~Lambotharan, and W.~H. Chin, ``Smart grid communications: overview
  of research challenges, solutions, and standardization activities,'' {\em
  Communications Surveys \& Tutorials, IEEE}, vol.~15, no.~1, pp.~21--38, 2013.

\bibitem{palensky2011demand}
P.~Palensky and D.~Dietrich, ``Demand side management: Demand response,
  intelligent energy systems, and smart loads,'' {\em Industrial Informatics,
  IEEE Transactions on}, vol.~7, no.~3, pp.~381--388, 2011.

\bibitem{mohsenian2010autonomous}
A.-H. Mohsenian-Rad, V.~W. Wong, J.~Jatskevich, R.~Schober, and A.~Leon-Garcia,
  ``Autonomous demand-side management based on game-theoretic energy
  consumption scheduling for the future smart grid,'' {\em Smart Grid, IEEE
  Transactions on}, vol.~1, no.~3, pp.~320--331, 2010.

\bibitem{Gellings2009}
C.~W. Gellings, {\em The smart grid: enabling energy efficiency and demand
  response}.
\newblock The Fairmont Press, Inc., 2009.

\bibitem{LeFloch2015Distributed}
C.~Le~Floch, F.~Belletti, S.~Saxena, A.~Bayen, and S.~Moura, ``Distributed
  optimal charging of electric vehicles for demand response and load shaping,''
  in {\em 2015 IEEE 54th Annual Conference on Decision and Control (CDC)},
  2015.

\bibitem{Lund2008}
H.~Lund and W.~Kempton, ``Integration of renewable energy into the transport
  and electricity sectors through {V2G},'' {\em Energy policy}, vol.~36, no.~9,
  pp.~3578--3587, 2008.

\bibitem{richardson2013electric}
D.~B. Richardson, ``Electric vehicles and the electric grid: A review of
  modeling approaches, impacts, and renewable energy integration,'' {\em
  Renewable and Sustainable Energy Reviews}, vol.~19, pp.~247--254, 2013.

\bibitem{kempton2005vehicle}
W.~Kempton and J.~Tomi{\'c}, ``Vehicle-to-grid power fundamentals: Calculating
  capacity and net revenue,'' {\em Journal of power sources}, vol.~144, no.~1,
  pp.~268--279, 2005.

\bibitem{Langton2013}
A.~Langton and N.~Crisostomo, ``Vehicle-grid integration: A vision for
  zero-emission transportation interconnected throughout california's
  electricity system,'' tech. rep., California Public Utilities Commission.

\bibitem{delfino2014multilevel}
F.~Delfino, R.~Minciardi, F.~Pampararo, and M.~Robba, ``A multilevel approach
  for the optimal control of distributed energy resources and storage,'' {\em
  Smart Grid, IEEE Transactions on}, vol.~5, no.~4, pp.~2155--2162, 2014.

\bibitem{shaaban2014real}
M.~F. Shaaban, M.~Ismail, E.~F. El-Saadany, and W.~Zhuang, ``Real-time pev
  charging/discharging coordination in smart distribution systems,'' {\em Smart
  Grid, IEEE Transactions on}, vol.~5, no.~4, pp.~1797--1807, 2014.

\bibitem{lopes2011integration}
J.~A.~P. Lopes, F.~J. Soares, and P.~M.~R. Almeida, ``Integration of electric
  vehicles in the electric power system,'' {\em Proceedings of the IEEE},
  vol.~99, no.~1, pp.~168--183, 2011.

\bibitem{mou2015decentralized}
Y.~Mou, H.~Xing, Z.~Lin, and M.~Fu, ``Decentralized optimal demand-side
  management for phev charging in a smart grid,'' {\em Smart Grid, IEEE
  Transactions on}, vol.~6, no.~2, pp.~726--736, 2015.

\bibitem{riverso2015plug}
S.~Riverso, F.~Sarzo, and G.~Ferrari-Trecate, ``Plug-and-play voltage and
  frequency control of islanded microgrids with meshed topology,'' {\em Smart
  Grid, IEEE Transactions on}, vol.~6, no.~3, pp.~1176--1184, 2015.

\bibitem{Melanie}
M.~N. Zeilinger, Y.~Pu, S.~Riverso, G.~Ferrari-Trecate, and C.~N. Jones, ``Plug
  and play distributed model predictive control based on distributed invariance
  and optimization,'' in {\em 52nd IEEE Annual Conference on Decision and
  Control, 2013}, pp.~5770--5776.

\bibitem{Stoustrup2009}
J.~Stoustrup, ``Plug \& play control: Control technology towards new
  challenges,'' {\em European Journal of Control}, vol.~15, no.~3,
  pp.~311--330, 2009.

\bibitem{bansal2014plug}
S.~Bansal, M.~N. Zeilinger, and C.~J. Tomlin, ``Plug-and-play model predictive
  control for electric vehicle charging and voltage control in smart grids,''
  in {\em 2014 IEEE 53rd Annual Conference on Decision and Control (CDC)},
  pp.~5894--5900, IEEE, 2014.

\bibitem{farivar2013branch}
M.~Farivar and S.~H. Low, ``Branch flow model: Relaxations and
  convexification---part i,'' {\em Power Systems, IEEE Transactions on},
  vol.~28, no.~3, pp.~2554--2564, 2013.

\bibitem{he2012optimal}
Y.~He, B.~Venkatesh, and L.~Guan, ``Optimal scheduling for charging and
  discharging of electric vehicles,'' {\em IEEE Transactions on Smart Grid},
  vol.~3, no.~3, pp.~1095--1105, 2012.

\bibitem{Hu2014Coordinated}
J.~Hu, S.~You, M.~Lind, and J.~Ostergaard, ``Coordinated charging of electric
  vehicles for congestion prevention in the distribution grid,'' {\em IEEE
  Transactions on Smart Grid}, vol.~5, no.~2, pp.~703--711, 2014.

\bibitem{Baran1989}
M.~E. Baran and F.~F. Wu, ``Optimal sizing of capacitors placed on a radial
  distribution system,'' {\em IEEE Trans. on Power Delivery}, vol.~4, no.~1,
  pp.~735--743, 1989.

\bibitem{Farivar2013}
M.~Farivar, L.~Chen, and S.~Low, ``Equilibrium and dynamics of local voltage
  control in distribution systems,'' in {\em 52nd IEEE Annual Conference on
  Decision and Control (CDC), 2013}, pp.~4329--4334.

\bibitem{Baran1989_2}
M.~E. Baran and F.~F. Wu, ``Optimal capacitor placement on radial distribution
  systems,'' {\em IEEE Trans. on Power Delivery}, vol.~4, no.~1, pp.~725--734,
  1989.

\bibitem{Bansai2014Plug}
S.~Bansal, M.~Zeilinger, and C.~Tomlin, ``Plug-and-play model predictive
  control for electric vehicle charging and voltage control in smart grids,''
  in {\em Decision and Control (CDC), 2014 IEEE 53rd Annual Conference on},
  pp.~5894--5900, Dec 2014.

\bibitem{Riverso}
S.~Riverso, M.~Farina, and G.~Ferrari-Trecate, ``Plug-and-play decentralized
  model predictive control for linear systems,'' {\em IEEE Trans. on Automatic
  Control}, vol.~58, no.~10, pp.~2608--2614, 2013.

\bibitem{Farivar2012}
M.~Farivar, C.~R. Clarke, S.~H. Low, and K.~M. Chandy, ``Inverter {VAR} control
  for distribution systems with renewables,'' in {\em IEEE International
  Conference on Smart Grid Communications, 2011}, pp.~457--462.

\end{thebibliography}

\end{document}